\documentclass[ prb,preprint, superscriptaddress]{revtex4}
\usepackage{graphicx}
\usepackage{dcolumn}
\usepackage{bm}
\usepackage{amsmath,amssymb,amsthm}

\newtheorem{theorem}{Theorem}[section]
\newtheorem{lemma}{Lemma}[section]

\begin{document}

\title{Efficient stochastic sampling of first-passage times with applications to self-assembly simulations}

\author{Navodit Misra}
\affiliation{Department of Physics,
Carnegie Mellon University, Pittsburgh, PA 15213}
\author{Russell Schwartz}
\affiliation{Department of Biological Sciences,
Carnegie Mellon University, Pittsburgh, PA 15213}
\begin{abstract}

Models of reaction chemistry based on the stochastic simulation
algorithm (SSA) have become a crucial tool for simulating complicated
biological reaction networks due to their ability to handle extremely
complicated networks and to represent noise in small-scale chemistry.
These methods can, however, become highly inefficient for stiff
reaction systems, those in which different reaction channels operate
on widely varying time scales. In this paper, we develop two methods
for accelerating sampling in SSA models: an exact method and a scheme
allowing for sampling accuracy up to any arbitrary error bound.  Both
methods depend on analysis of the eigenvalues of continuous time
Markov models that define the behavior of the SSA.  We show how each
can be applied to accelerate sampling within known Markov models or to
sub-graphs discovered automatically during execution.  We demonstrate
these methods for two applications of sampling in stiff SSAs that are
important for modeling self-assembly reactions: sampling breakage
times for multiply-connected bond networks and sampling assembly times
for multi-subunit nucleation reactions.  We show theoretically and
empirically that our eigenvalue methods provide substantially reduced
sampling times for a large class of models used in simulating
self-assembly. These techniques are also likely to have broader use
in accelerating SSA models so as to apply them to systems and
parameter ranges that are currently computationally intractable.
\end{abstract}

\maketitle

\section{Introduction}

Stochastic simulation methods have become increasingly widespread as a
means of simulating and analyzing biochemical reaction
kinetics\cite{Rao02}. The chemical master equation, which governs the
reaction kinetics for well-mixed systems, forms the basis for the
stochastic simulation algorithm (SSA), proposed by Gillespie
\cite{Gill76,Gill77}. SSA models a reaction system as a Continuous
Time Markov Model (CTMM) in which states of the system are defined by
counts of reactants present at a given point in time and transitions
between states correspond to individual reaction events. This SSA
approach is valuable in part because it provides a model of reaction
noise, which can become significant for reaction networks on cellular
scales~\cite{Ark97}. Furthermore, SSA models can provide significant
computational advantages over continuum models for networks
characterized by extremely large sets of possible reaction
intermediates.  The computational value of the SSA approach lies in the fact that
for a large class of networks, the random walk visits only a small
fraction of the state space before equilibrium is established.  As a
result, kinetics on complicated networks can be simulated ``on the fly,'' 
requiring explicit construction of the CTMM network only in the immediate
vicinity of those states visited on a given trajectory. This property is
an essential requirement for any feasible simulation algorithm, since
the size of the state space describing the master equation is
astronomical even for modest system sizes. Successful applications of
SSA include gene regulatory networks~\cite{Ark97} and self-assembly of
complicated structures, such as virus capsids \cite{Ber94,Zhang06}.
Furthermore, the SSA approach has now been adopted by several
approaches for whole-cell modeling~\cite{StochSim,E-Cell3} and
modeling generic complex reaction networks~\cite{BioNetGen,Moleculizer}.

The relaxation time of the SSA can, however, be extremely sensitive to
the transition rates controlling the reaction kinetics. A pure SSA
model has difficulty with stiff reaction systems, i.e., those where
important events occur in parallel on very different time
scales.  In such cases, a simulation can become bogged down by
sampling fast events to the exclusion of the slow events.  Hybrid
discrete/stochastic models~\cite{Gill01,Hasel02,Cer02} can
resolve this problem in some domains, but not when the fast reactions
make use of too many intermediates to allow them to be modeled
continuously. One important example of such a stiff reaction system
is the breaking of bond networks, where individual bonds may break and
repair repeatedly before a sufficiently large bond group is broken to
fracture the network.  Another form of stiff SSA network occurs near
the critical concentration of a self-assembly system, where high-order
nucleation events can be orders of magnitude slower than individual
binding reactions.  In these stiff systems, an SSA model can become
``trapped'' for many steps in a small subset of the state space,
resulting in negligible simulation progress for long periods of time.

\begin{figure}
\includegraphics[scale=0.8]{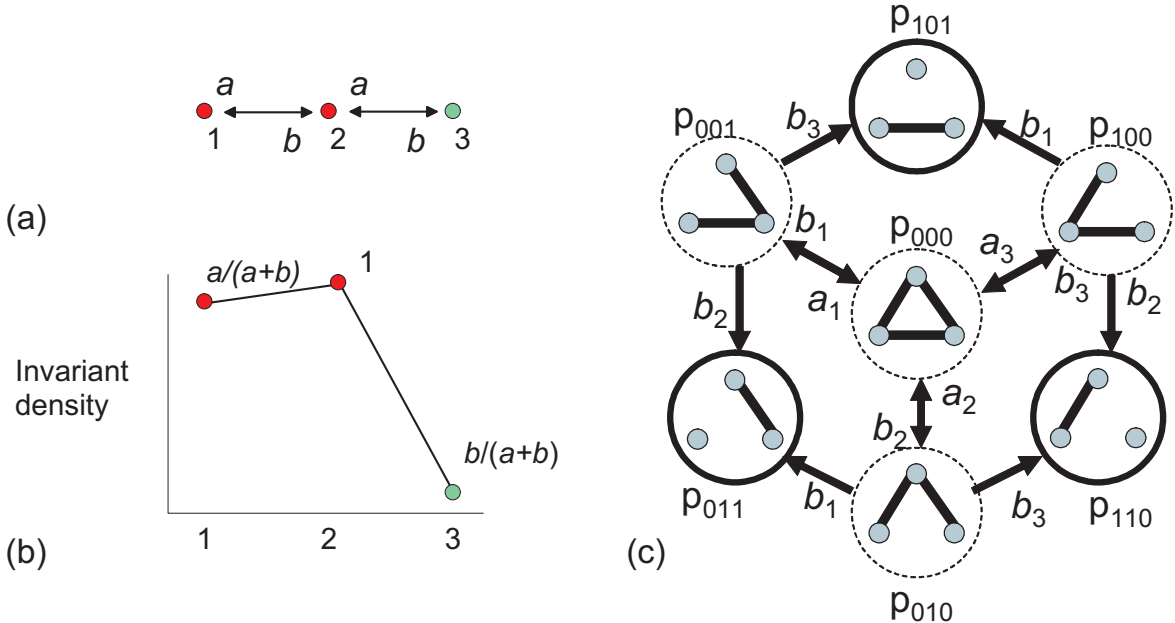}
\caption{Illustration of trapped subgraphs in SSA models. (a) A simple
CTMM on a 3-path with transition rates $a$ and $b$. (b) The
probability landscape for the model. SSA is slow whenever the
invariant density for the corresponding Markov chain is irregular. Here, the SSA takes
$O(a/b)$ steps to reach vertex 3. (c) CTMM model of a
trimer assembly system with three subunits. Graph of possible
configurations joined by reaction rates.  States in which the trimer
is broken are surrounded by solid lines and others by dashed lines.}
\label{SSA}
\end{figure}

To understand these ``trapped'' systems, it is useful to consider the
graph theoretic representation of the SSA method.  An SSA model is
represented by a graph in which each node corresponds to one possible
state of the full model.  Edges connect nodes whose states can be
reached from one another by a single reaction event, e.g., two
molecules binding to one another.  At each simulation step, SSA
considers only the immediate neighbors of the current state.  As a
result, the simulation is prone to traps that can result from
irregularities in the invariant density of the embedded Markov chain
(EMC) implemented by SSA for a given CTMM.  For example, consider a
3-state CTMM represented by a simple path (Fig.~\ref{SSA}(a)), where
the backward transition rate $a$ is much larger than the forward
transition rate $b$.  The average number of SSA steps to reach state 3
from initial state 1 is $O(a/b)$ because once SSA visits state 2, it
will jump to state 3 only a $b/a$ fraction of the time. Nodes 1 and 2
collectively define a trapped subgraph from which the model must
escape. In general, for an $N$-path, where each forward rate $b$ is
smaller than the backward rate $a$, SSA takes $O(a/b)^{N-2}$ steps to
traverse the path (see Theorem \ref{SSAcon} for an analogous
problem). One way such a trapped subgraph can arise in a physical
system is through models of the breakage of bond networks.
Fig.~\ref{SSA}(c) shows the graph arising from a model of the breakage
of a three-cycle bond network, which behaves similarly to the 3-state
CTMM by establishing a trapped inner graph of four states --- the
unbroken state and three states with a single broken bond --- from
which the model must escape to reach any broken network state.  We can
alternatively understand the trapping problem in terms of a
probability landscape view of a reaction system.  The SSA is sluggish
whenever its equilibrium landscape is irregular, consisting of valleys
and hills. The broader and deeper these are, the slower SSA becomes.

To overcome the presence of traps or landscape irregularity, we
propose two \emph{non-local} simulation algorithms that rely on the
spectral decomposition of the Kolmogorov matrix (for a CTMM) or the
transition matrix (for the Embedded Markov Chain (EMC)). These
eigenvalues and their associated eigenvectors describe global modes of
relaxation of the full graph or any of its sub-graphs. Since
eigenvalues are global properties of a graph, spectral methods are
much less sensitive to local landscape traps.  These methods can be
applied to quickly sample first passage times on small CTMM graphs
such as those in Fig.~\ref{SSA} or to sample escape times from trapped
subgraphs when the full model is prohibitively large.

Previous attempts at simulating rare events include the Forward Flux
Sampling (FFS) technique of Allen {\it et al.}\cite{All05} and related
methods \cite{Crooks01,Del02}. The approach breaks a rare event into a
series of relatively more probable stages and uses estimates of
waiting times for the successive stages to develop an aggregate
transition rate for the full event.  This aggregate rate can then be
used to approximate the first passage time density as a single
exponential random variable.  However, while the exponential tail
dominates the density for stiff systems and is therefore a highly
accurate approximation in many cases, the true probability density has
a peak at short times followed by a mixed exponential tail. The
methods developed in the present work, by contrast to the FFS-like
methods, sample first passage times from the entire density to within
arbitrary an error bound.  Recently, another method called the
slow-scale SSA was proposed by Cao {\it et al.}\cite{Cao04,Cao05},
which relies on a technique called the Partial Equilibrium
Approximation (PEA).  PEA essentially assumes that the set of fast
reactions are always in equilibrium and the method approximates
transition rates between slowly varying reactant species by their
expected value in the partial equilibrium state. While these methods
can provide significant benefits for some CTMMs, there are several
limitations in using PEA or similar approximations for arbitrary
graphs. First, a clear distinction between fast and slow species may
not be obvious in a given problem.  For example, in rule-based
simulation of bond networks, stiffness is built in through the
association/dissociation rates of individual bonds rather than being
species dependent. Secondly, these methods always need to be
supplemented with approximations involved in computing the mean values
of the reaction propensities. Furthermore, PEA will be inaccurate
whenever fluctuations in the reaction propensities within the partial
equilibrium state are comparable to their mean values.

The goal of the present work is to develop efficient methods for some
important classes of stiff SSA model for which the above techniques
are unsuitable, with a particular emphasis on models important to
simulations of self-assembly reactions.  The methods proposed in this
paper can be applied to Markov processes on arbitrary
graphs. Furthermore, they can be made accurate to within arbitrary
error bounds.  The remainder of this paper is organized as follows:
Section \ref{notation} sets up some basic notation and a description
of the sampling problem for general CTMM. In Section \ref{spec1} we
introduce a spectral method which relies on the eigen decomposition of
the master equation describing the CTMM.  We use a complete spectral
decomposition of the first passage time density and rejection sampling
to return sample first passage times for arbitrary CTMMs. In section
\ref{spec2} we introduce another spectral method which works as a
hybrid between the purely local SSA and the completely nonlocal Master
Equation method. The latter method proceeds by adaptively constructing
a basis in which to simulate the Markov chain until the system state
has relaxed to its slowest eigenvector. If first-passage out of the
trapped subgraphs does not occur by that time, we use the appropriate
eigenvalue to sample the time to first-passage as an exponential
random variable. In section \ref{AD} we introduce a method for
automated discovery of trapped regions in stiff Markov models. This
technique allows efficient implementation of spectral methods for
large state spaces by isolating regions repeatedly visited by a given
random trajectory and using spectral sampling to escape any such
subgraph. In section \ref{BondNet} we present theoretical results on
the time complexity of SSA for bond networks followed by experiments
on some special classes of bond networks to compare the simulation
efficiency of each method discussed. In Section \ref{NuclLim} we
evaluate the automated discovery variants of the method by applying
them to models of a nucleation-limited assembly system with a state
space too large to explicitly construct. Section \ref{discussion}
concludes the paper with a discussion of results and directions for
future research.

\section{Theory}
 \subsection{\label{notation}The chemical master equation and the stochastic simulation algorithm}
The SSA identifies reaction kinetics for networks of biochemical
subunits as a Markov process governed by an appropriate
\emph{Chapman-Kolmogorov} equation or, equivalently, its differential
version - the master equation. Let $S=\{1,2,\ldots,N_S\}$ be the state space
 for the CTMM, each node representing a possible state for the
simulated system. The time evolution of probability densities is
governed by a \emph{Kolmogorov} matrix $W$, which specifies the
transition rates $W_{nm}$ from the state $m$ to $n$.
\begin{equation}
\frac{dp_{n}}{dt}=\sum_{m \in S} W_{nm}p_{m}(t) - W_{mn} p_{n}(t)
\end{equation}
where, $p_{n}(t)$ denotes the probability to be in state $n$ at time
$t$. The matrix elements $W_{nm}$ satisfy two necessary conditions:
\begin{enumerate}
\item $W_{nm} \geq 0$ for $n \neq m$.
\item $\sum_{m} W_{nm} = 0$.
\end{enumerate}
Under these conditions, it is well known that the matrix has a steady
state solution $|\Pi\rangle = \sum_{n}\pi_{n}|n\rangle$ that is an
eigenvector of $W$ with eigenvalue zero and that all initial
distributions relax to $|\Pi\rangle$ in the limit of long
times~\cite{vanKampen}. In addition, we will require $W$ to satisfy
the \emph{detailed balance} condition, which states that at
equilibrium, the sum of probability current exchanged between any pair
of states $(n,m)$ is zero, i.e., $W_{nm}\pi_{m}=W_{mn}\pi_{n}$. This
in turn allows one to define a scalar product on the state space such
that $W$ is \emph{self-adjoint}:
\begin{equation}
\langle n|m \rangle \equiv \delta_{nm}\frac{1}{\pi_{m}}
\end{equation}
This condition ensures that we can construct an orthogonal
\emph{eigenbasis} and compute time evolved versions of any given
initial probability distribution using spectral decomposition.

\subsection{\label{spec1}Spectral Sampling 1: Master Equation approach}
Given a Kolmogorov matrix $W$ on a state space $S$ and an arbitrary initial
state $i \in V \subset S$, the first-passage time $T_{F}(i)$ is a random
variable which gives the time at which the trajectory first reaches
any state in some subset of the state space $F = S-V$.
The standard method of solving a first passage problem is to set up
the master equation for $V$ with an absorbing boundary over $F$
(zero Dirichlet boundary condition)~\cite{vanKampen}.  Let $P_{V}$ be a projection operator onto the subspace
$V$ and let $N$ be the cardinality of $V$. Then,
$M=P_{V}WP_{V}$ is the effective Kolmogorov matrix that governs time
evolution over $V$. From detailed balance, $M$ is self-adjoint
over $L^{2}_{\pi^{-1}}$. Hence, the eigenvectors of $M$ form a complete
basis $\{|\psi_{\alpha}\rangle\}$. A consequence of the
spectral theorem is the completeness relation for the properly
normalized eigenbasis, i.e.,
$\langle\psi_{\alpha}|\psi_{\beta}\rangle=\delta_{\alpha\beta}$
. Given any vector $|\eta\rangle$:
\begin{equation}
|\eta\rangle = \sum_{\alpha=1}^{N} |\psi_{\alpha}\rangle\langle\psi_{\alpha}|\eta\rangle
\end{equation}

\subsubsection{Spectral decomposition of the first-passage time distribution}
 In terms of the vertex set basis, the completeness relation over $L^{2}_{\pi^{-1}}$ is $I=\sum_{n\in V}P_{n}$, where $P_{n}=\pi_{n}|n\rangle\langle n|$ is the projector onto vertex state $|n\rangle$. Given an initial probability density $p_{n}(t=0)=\delta_{ni}$ the probability for state $n\in V$ evolves as:
\begin{eqnarray}
p_{n}(t)|n\rangle&=& P_{n} e^{t M} |i\rangle= \pi_{n}|n\rangle\langle n| \sum_{\alpha=0}^{N} \langle \psi_{\alpha}|i\rangle e^{-\lambda_{\alpha}t}| \psi_{\alpha}\rangle \nonumber\\
\Rightarrow p_{n}(t)&=&\sum_{\alpha=1}^{N}\pi_{n}\psi_{\alpha,n}\psi_{\alpha,i}\exp[-\lambda_{\alpha}t]
\end{eqnarray}
The transition to an element $f\in F$ outside of $V$ , is governed by the following equation:
\begin{eqnarray}
\frac{dp_{f}}{dt}&=& \pi_{f}\langle f|(W-M)\sum_{n\in V}p_{n}(t)|n\rangle\nonumber\\
&=& \sum_{\alpha=1}^{N}c_{\alpha,f}\exp[-\lambda_{\alpha}t]
\end{eqnarray}
The probability for a first passage to the state $f$ between $t$ and $t + dt$ is hence  given by $\rho( T_{f}=t)dt = \sum_{\alpha=1}^{N}c_{\alpha,f} e^{-\lambda_{\alpha}t}dt$.

\subsubsection{Exact sampling for the first-passage time distribution}
In this section we describe a method for returning a sample time from
the computed first-passage density $\rho(t)=\sum_{i=1}^{N}c_{i}
e^{-\lambda_{i}t}$ to any state $f\in F$. A general method for sampling
from complicated distributions is to use the method of rejection
sampling, which first chooses a random variable from a convenient
envelope density and accepts or rejects the sample based on a second
 random sample that depends on the tightness of the envelope fit. The
  rejection rate is low if the envelope curve closely approximates the
 given curve. A simple envelope curve is provided by a pure exponential 
 of the most slowly decaying eigenvalue, with a coefficient equal to the
  sum of all positive terms $\sum_{c_{i}>0} c_{i}$ in the computed density 
$\rho(t)$. However, there is no guarantee that the rejected part is small.
 Since each eigen mode encloses an area $c_{i}/\lambda_{i}$, cancellations
 between near-degenerate eigenvalues can in principle lead to a high
 rejection ratio. We therefore present a method for choosing an envelope
 curve $g(t)$ which eliminates these cancellations. Furthermore, in section \ref{MEnet} we show 
 that the envelope curve is exact for bond networks generated by cycle graphs
$C_{N}$. We sample from $g(t)$ using a decomposition into a discrete
mixture of densities $f_{\alpha}(t) = d_{\alpha}(e^{-\lambda_{\alpha}
t}-e^{-\lambda_{\alpha + 1}t})$ and an efficient rejection step. Here $d_{\alpha}$ are constants, one for each component $f_{\alpha}$ of the envelope curve $g(t)$. The next theorem proves that the density $f_{\alpha}(t)$ can be sampled efficiently using a rejection method.

\begin{figure}
\begin{center}
\includegraphics[scale=0.9]{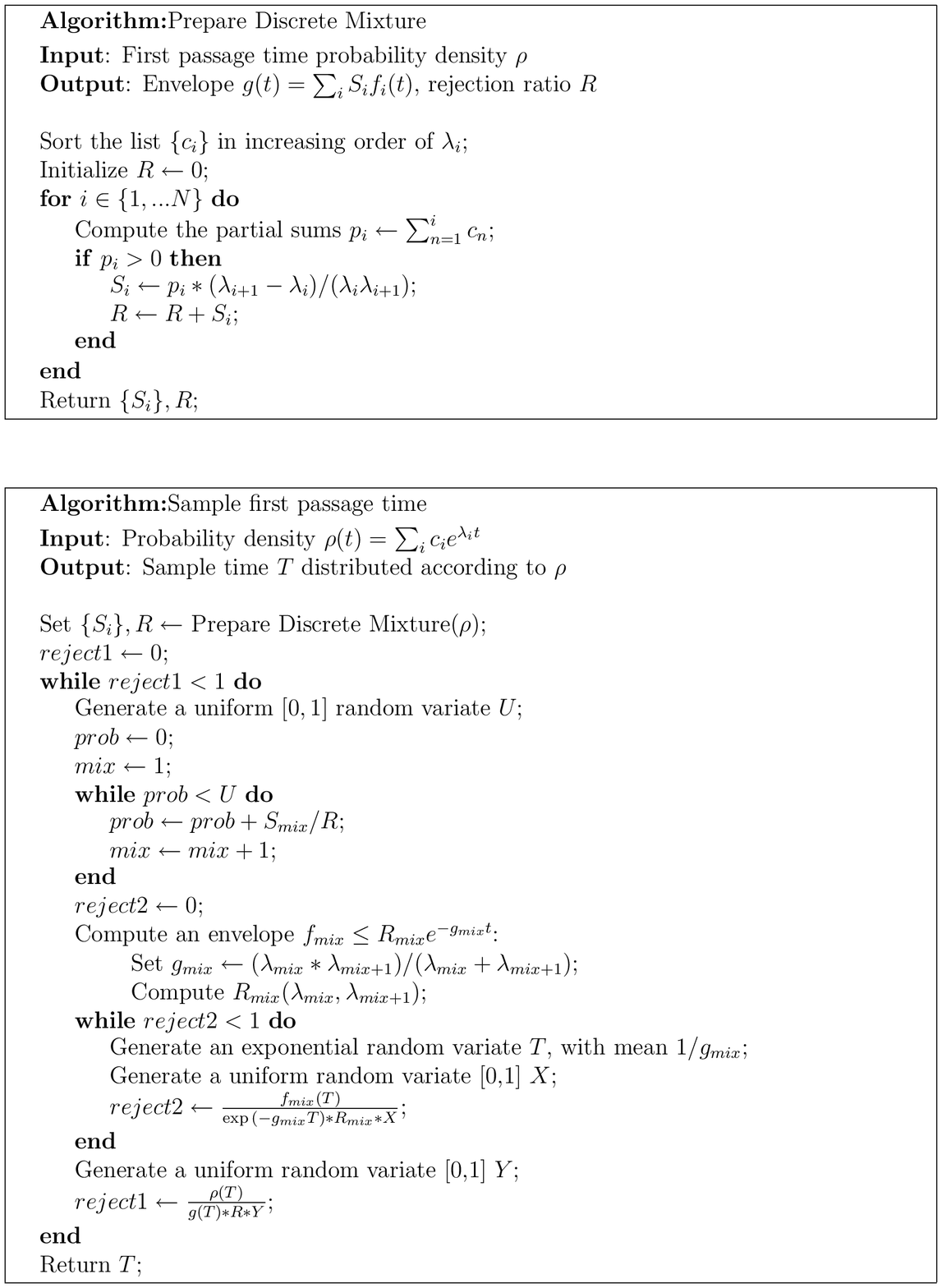}
\caption{Pseudocode for spectral method 1}
\end{center}
\end{figure}

\begin{theorem}\label{RRatio}
The expected rejection ratio for $f_{\alpha}(t)$ is bounded from above by 1.5.
\end{theorem}
\begin{proof}
We will use a simple exponential $h_{\alpha}(t)= C \lambda\exp[-\lambda t]$ as
the envelope function. In order to minimize $C$ we choose $h_{\alpha}(t)$ such
that \begin{equation}h_{\alpha}(t_{*})=f_{\alpha}(t_{*})\end{equation} and
\begin{equation}\left.\frac{dh_{\alpha}}{dt}\right|_{t=t_{*}}=\left.\frac{df_{\alpha}}{dt}\right|_{t=t_{*}}\end{equation}
where $t_{*}$ is defined implicitly by the condition
\begin{equation}\left.\frac{d^{2}f_{\alpha}}{dt^{2}}\right|_{t=t_{*}}=0
\end{equation} These constraints yield a unique solution
$t_{*}=2\frac{ln(\lambda_{\alpha+1}/\lambda_{\alpha})}{\Delta\lambda}$. Since
$\frac{d^{2}f_{\alpha}}{dt^{2}} > 0$ for $t>t_{*}$, the slope of
$\ln[f_{\alpha}]$ monotonically increases to $-\lambda_{\alpha}$ as
$t\rightarrow\infty$. The corresponding envelope rate then satisfies $\lambda=
\frac{\lambda_{\alpha}\lambda_{\alpha+1}}{\lambda_{\alpha}+
\lambda_{\alpha+1}} \leq \lambda_{\alpha}$. The rejection ratio is
given by:
\begin{eqnarray}
C&=&\frac{1}{\lambda}\frac{\lambda_{\alpha}\lambda_{\alpha+1}}{\lambda_{\alpha+1}-\lambda_{\alpha}}\exp[(\lambda-\lambda_{\alpha})t_{*}]\left(1-\left(\frac{\lambda_{\alpha}}{\lambda_{\alpha+1}}\right)^{2}\right)\nonumber\\
&=& \left[\frac{\lambda_{\alpha}+
\lambda_{\alpha+1}}{\lambda_{\alpha+1}}\right]^{2}\exp\left[-\frac{\lambda_{\alpha}^{2}}{\lambda_{\alpha}+
\lambda_{\alpha+1}}t_{*}\right]\leq
4\exp{\left[2\frac{x^{2}\ln{[x]}}{(1-x^{2})}\right]}
\end{eqnarray}
where $x=\lambda_{\alpha}/\lambda_{\alpha+1}\in (0,1)$. To upper-bound $C$, note that the exponent increases monotonically with $x$ and its maximum is $\lim_{x\rightarrow 1-}(2x^{2}\ln{[x]})/(1-x^{2}) = -1$. This bound finally gives us $C\leq 4/e \approx 1.47$.
\end{proof}

As a final comment, we note that for general graphs the average time complexity of this algorithm is dominated by the computation of the eigenvectors and eigenvalues, which gives us the following theorem:
\begin{theorem}
The average time complexity for spectral decomposition of the master equation is $O(N^{3})$ for a graph of $N$ vertices~\cite{Press}.
\end{theorem}

\subsection{\label{spec2}Spectral Sampling 2: Modified embedded Markov chain method}
The efficiency of the SSA is dependent on the relaxation time of the
embedded Markov chain (EMC). We use this observation to modify the
basis in which the EMC is simulated. The standard method of executing
a random walk is to consider the transition between adjacent states,
each of which is localized at a vertex of the CTMM. However, correct
simulation only requires that these states form a basis, not that they
are orthogonal. If we can choose a set of states which are
increasingly likely to appear during the simulation of the Markov
chain, we are unlikely to make repeated visits to the same state. In
order to identify such a basis starting from an initial state
$|i\rangle$, we first identify the transition matrix for an embedded
Markov chain that correctly describes the given CTMM. Consider the
vertex set $V=\{1,2, \ldots, N\}$ and the basis constructed from $V$,
$B=\{|1\rangle,|2\rangle, \ldots, |N\rangle\}$.  At any given time
$t$, let the state of the time-evolved Markov chain be
$|\psi(t)\rangle=\sum_{i=1}^{N}\psi_{i}|i\rangle$. Let $V_{t}=\{i \in
V|\psi_{i} \neq 0\}$ be the vertex subset populated by the current
state vector. We construct the EMC for the subgraph induced by $V_{t}$
at each step of the algorithm. Given the projection of the Kolmogorov matrix $M$ over the vertex set $V$, choose
$r=max(-M_{ii}|\psi_{i} \neq 0)$ to be the effective rate of
transition to the next state and choose an exponentially distributed
random time step $\tau$ with mean waiting time $1/r$. Then,
$L_{t}=-(1/r)*M $ is the Laplacian governing the EMC and $Q_{t} = I -
L_{t}$ is the effective transition matrix at that time step. The next
state vector is chosen to be $|\phi\rangle =
Q_{t}|\psi(t)\rangle$. The reason for choosing this particular value
of $r$ is to ensure that no term in $Q_{t}$ becomes negative, a
necessary condition for a transition matrix.
\begin{figure}
\includegraphics[scale=.8]{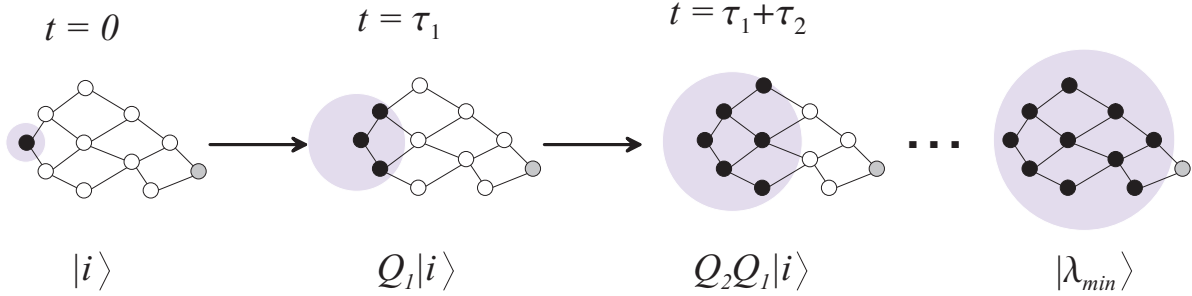}
\caption{Schematic of the EMC-based spectral method. Vertices in black are the currently occupied nodes. Simulation advances the system state as a discrete mixture until such time as the state has relaxed to its slowest eigenstate $|\lambda_{min}\rangle$. At each step, direct transitions to the absorbing vertex (grey) are computed according to the Kolmogorov matrix.}
\end{figure}
\begin{theorem}
The choice of next state is consistent with the master equation governing the CTMM.
\end{theorem}
\begin{proof} Rewrite the master equation in terms of the $\{|\psi\rangle,|\phi\rangle\}$ basis (where the other $N-2$ linearly independent basis vectors can be chosen arbitrarily):
\begin{eqnarray}
\frac{d|\psi\rangle}{dt}&=& r\left( I + \frac{1}{r}M - I\right)|\psi\rangle \nonumber\\
&=& r\left(Q_{t} - I\right)|\psi\rangle \nonumber\\
&=& r\left( |\phi\rangle - |\psi\rangle \right)
\label{EMCeq}
\end{eqnarray}
Since there is a unique decomposition for any vector in terms of a linearly independent basis set, Eq. \ref{EMCeq} proves that starting from $|\psi\rangle$ the next state is uniquely determined to be $|\phi\rangle$.
\end{proof} 
\begin{figure}
\includegraphics[scale=0.9]{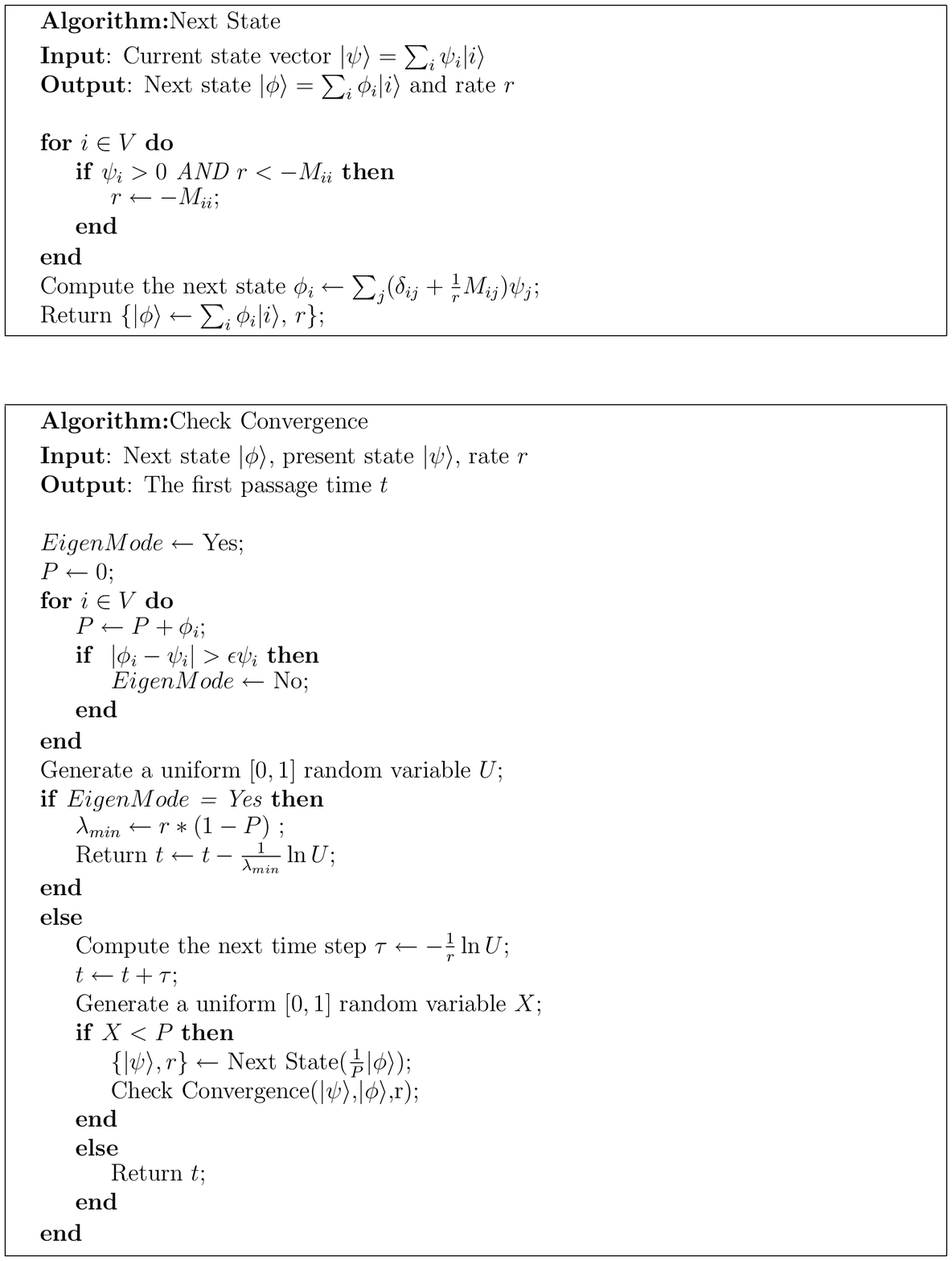}
\caption{Pseudocode for Spectral method 2}
\end{figure}

  In general, the next state $|\phi\rangle$ will have a total
probability $P = \sum_{i}\phi_{i} \le 1$, due to possible transitions
out of the subgraph. We check if that is the case by generating a
$[0,1]$ random variable $X$ to compare with $P$. If $X<P$, the next
state is still trapped inside the subgraph and we normalize it
as $\psi(t+\tau)=1/P|\phi\rangle$. $|\psi(t+\tau)\rangle$ is used to
generate the next state in the simulation.  This sequence will
continue until the state has relaxed to its slowest eigen vector
$|\lambda_{min}\rangle$, such that $M|\lambda_{min}\rangle =
-\lambda_{min}|\lambda_{min}\rangle$ to within a user-defined relative
error $\epsilon$. Once that state is achieved, we just need one more
exponentially distributed random sample time $\tau$ with mean
$1/\lambda_{min}$ to escape the network.

SSA chooses a stochastic trajectory by sampling both the next neighbor
and the time for the next step at random. The EMC method, on the other
hand, evolves deterministically in our modified basis and only the time
between transitions is stochastic. At each time step, transition to the
absorbing boundary states is governed by the matrix elements
connecting each of the transient states to the absorbing boundary. The
advantage of such an approach is that it allows us to automatically
compute the most slowly decaying eigenvector during the simulation. For
completeness, we note the following result:
\begin{theorem} For a graph of degree bounded by $d$ and $V$ of cardinality $N$, each step of this algorithm takes $O(N*d)$ time.\end{theorem}

\subsection{\label{AD}Automated discovery of trapped subgraphs} 

As previously mentioned, stiffness in Markov model graphs results from
repeated visits by a typical random trajectory to a small subset of
vertices of the entire graph. Since the performance of spectral
methods is sensitive to the size of the vertex set, it would prove
useful if we could somehow identify these ``trapped'' subgraphs for
stiff Markov models and apply spectral methods directly to those. In
this section we present one such method, which we call ``Automated
Discovery'' (AD) and which we show to be formally applicable to arbitrary
bounded-degree graphs.

Let there be a state space $S$ over which a CTMM is defined and
consider a subgraph $G(V,E)$ with vertex set $V \subset S$ and edge
set $E$. Starting from an initial state $i\in V$, we are interested in
the time $T_{F}(i)$ to first passage out of $V$. Consider the subgraph
$H_i$ induced by the vertex set $U_i\subseteq V$ visited by a random
trajectory executing the SSA random walk before it escapes $V$ and let
$N_i =|U_i|$ be the cardinality of $U_i$. If $T_{Fi}$ is the number of
steps a SSA random walk takes to escape $V$, then a Markov model will
show simulation stiffness whenever the expected values satisfy,
\begin{equation}\label{stiffAD}
E[N_i] << E[T_{Fi}]
\end{equation}
since this would imply certain vertices in $U_i$ are being visited
repeatedly. AD works by progressively sampling larger regions of $V$
until it identifies a subgraph $K_i$ induced by a vertex set
$W_i\subseteq V$ such that $U_i \subseteq W_i$. Once $K_i$ is
identified either of the spectral methods can be used directly over
$K_i$. The method will be efficient as long as $|W_i| \sim |U_i|$ and
the number of steps taken to identify $K_i$ is comparable to the
computational cost of using spectral sampling over $K_i$. Formally,
$H_i$ can be exactly discovered by repeatedly enlarging the discovered
graph to include the last vertex outside $K_i$ visited by the
trajectory. If spectral sampling for a graph of vertex set size $N$
works in time $f(N)$, this procedure would ensure that implementing
spectral sampling in conjunction with automated discovery takes
$O(N_i*f(N_i))$ steps. The stiffness condition (Eqn.~\ref{stiffAD})
would usually ensure that this procedure is still efficient. *B* Another
method for discovering the trapped subgraph would be to implement the
SSA random walk for a specified number of steps $S(N)$(depending on the size
of the vertex set $N$). Since eigenvalue methods are in general
$O(N^3)$, we can implement SSA until $S(N) \leq C*N^3$ for some constant $C$,
to discover the trapped subgraph $K$ and then use spectral sampling to escape the discovered graph.  This
alternative approach could be less efficient in some circumstances,
but would guarantee that the overhead for spectral sampling is no
more than a constant factor beyond that of the standard SSA.
Fig.~\ref{ADpcode} shows the pseudocode for implementing AD for a given graph by this method. 
The algorithm generates a sample trajectory using SSA till such time 
that the trajectory spends $O(N^3)$ steps within a trapped graph $K$ of vertex set cardinality $N=|K|$. 
Then either of the spectral methods described in section \ref{spec1} or \ref{spec2} are 
used to sample the first passage outside $K$, to a vertex $i$. In general the state of the system at 
the time of first passage outside $K$ will be a discrete probability mixture of more than one vertices.
 In such cases, the vertex $i$ is randomly selected in accordance with the appropriate probability weight.
The algorithm then resumes SSA execution over the enlarged graph $K\bigcup\{i\}$.*E*
 Further investigation is, however, required
to search for algorithms that may further improve the performance of
AD. In section \ref{latticeAD} we prove that
for at least one important class of graphs, namely models of
chemically reacting species, we can indeed reduce the time complexity
to its optimal value to within a constant factor, \emph{i.e.},
$O(f(N_i))$.

\begin{figure}
\includegraphics[scale=0.8]{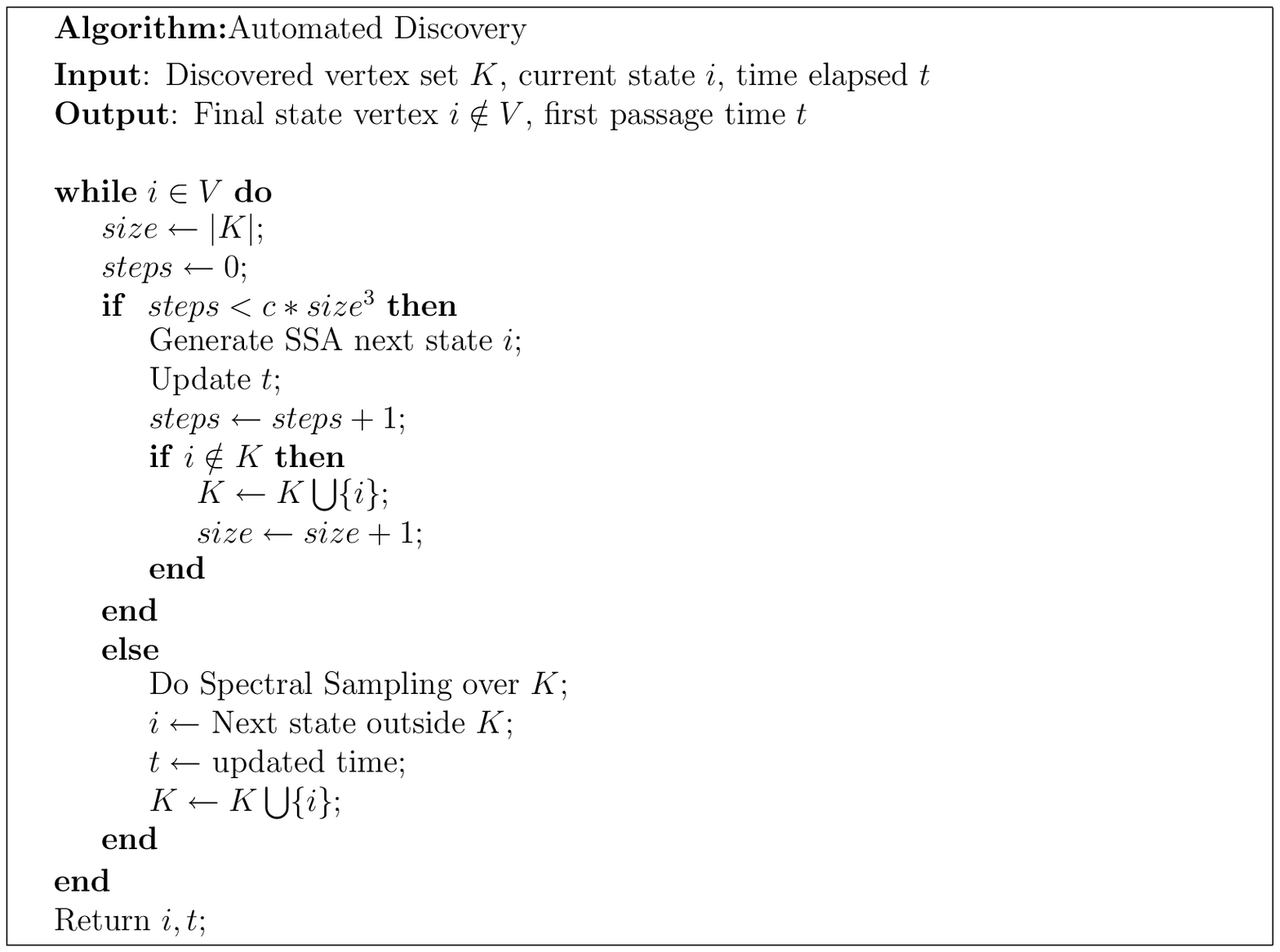}
\caption{Pseudocode for Automated Discovery}\label{ADpcode}
\end{figure}

\section{\label{BondNet}Application When the Subgraph is Known: Fracturing Bond Networks}

\subsection{Stiffness in SSA for bond networks} 

In order to validate the methods, we instantiate them for some
specific challenging systems.  We begin by demonstrating the non-AD
variants of the methods for the problem of sampling the time required
to break a network of bonds.  This problem is an example of a stiff
SSA on a generally small graph.  It is also of independent interest
because of its importance in modeling self-assembly processes on long
time scales.  Given such a system, we are interested here in the first
passage time to the subset of states corresponding to disconnected
graphs $V_{b} \subset S$. Since each bond can occur in two states,
intact or broken, a network of $d$ bonds can be represented as a
vertex on a unit hypercube in $d$ dimensions. The state space
generated by the bond network before it becomes disconnected will
usually be a truncated unit hypercube. An N-cycle $C_{N}$ generates
the simplest non-trivial example, where the absorbing boundary is
placed at all points on the hypercube at distance 2 from the
fully-connected state.  Fig.\ref{SSA}(c) illustrates this absorbing
boundary for $C_{3}$. Given a $d$-bond network, we will represent the
$\mu^{th}$ bond-breaking rate by $b_{\mu}$ and association or binding
rate by $a_{\mu}$. It is convenient to represent a vertex on this
hypercube by a binary $d$-tuple $\bm{i}= \{i_{d},\ldots i_{\mu},\ldots
i_{1}\}$, where $i_{\mu}=0$ implies that the $\mu^{th}$ bond is intact
(see Fig. \ref{SSA}(c) for the graph corresponding to a trimer).  From
here on, we will use the notation
$\hat{\bm{\mu}}=\{\delta_{d\mu},\ldots, \delta_{1\mu}\}$ for the
vector describing a state of the model with only the $\mu^{th}$ bond
broken. For such a graph, the time complexity of each SSA step is
$O(d)$. In the rest of this paper we will use this model of truncated
hypercubes to represent bond networks. Morris and
Sinclair\cite{Morr05} have proven that in the case of unweighted
graphs, a random walk on a hypercube truncated by a hyperplane relaxes
to equilibrium in polynomial time bounded by $O(d)^{9/2+\epsilon}$ for
any $\epsilon > 0$. However, as we have argued in the introduction,
the mean \emph{hitting time}, i.e., the number of random walk steps
between a pair of vertices, can be extremely sensitive to the
parameters governing the walk. We formalize this observation in the
Theorem~\ref{SSAcon} below, which bounds the expected number of SSA
steps before the network is disconnected. Let $r\equiv
Min(a_{\mu}/b_{\nu}|\mu ,\nu \in\{1,\ldots d\})$.
\begin{theorem}\label{SSAcon} The expected number of SSA steps required to break a $k$-connected network with $k>1$ and $r>1$ is $\Omega(r^{k-1})$.
\end{theorem}
 A detailed proof of the theorem is provided in the appendix.
Figs.~\ref{CNSSAsteps} and \ref{ZNSSAsteps} provide an empirical
demonstration of the theorem. Fig.~\ref{CNSSAsteps} analyzes the
number of steps required in 100 trials of the SSA algorithm for
simulating the breakage of a set of cycle graphs $C_N$ ranging in size
from three to seven.  Each model was examined using ratios of forward
to backward rate from 1 to 20 in increments of 1.  Breakage times for
the cycle graphs increase linearly with rate ratio, although they also
fall monotonically with cycle size (Fig.~\ref{CNSSAsteps}(a)).
Fig.~\ref{ZNSSAsteps} analyzes the number of steps required to break
k-connected hypercube graphs of dimensions
$k=\{2,3,4,5\}$. Fig.~\ref{ZNSSAsteps}(a) also shows that the slope of
a log-log plot approaches the predicted exponent $k-1$.
Fig. \ref{CNSSAsteps}(b) and \ref{ZNSSAsteps}(b) suggest why a
spectral approach might be effective --- as the reaction rate increases,
steps to first passage behave more like a geometric random variable
(as mean $\rightarrow \infty$, standard deviation $\rightarrow$ mean),
as expected for a slowly decaying eigen mode of the transition matrix.
More detailed explanations of the simulation protocol for these
figures is provided in section \ref{netmod}.

\begin{figure}
\centerline{\includegraphics[scale=0.700]{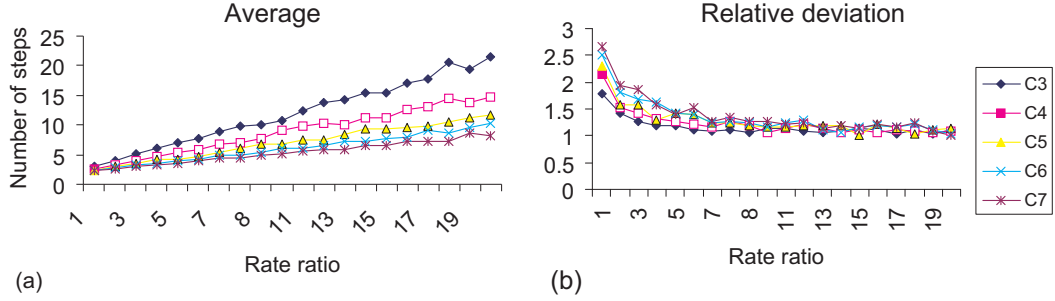}}
\caption{Number of SSA steps until first passage for the network generated by an N-cycle $C_{N}$. (a) Average number of steps $\langle s \rangle$ (b) Relative deviation $ \frac{\surd\langle \delta s^{2} \rangle}{\langle s \rangle}$}
\label{CNSSAsteps}
\end{figure}
\begin{figure}
\centerline{\includegraphics[scale=0.700]{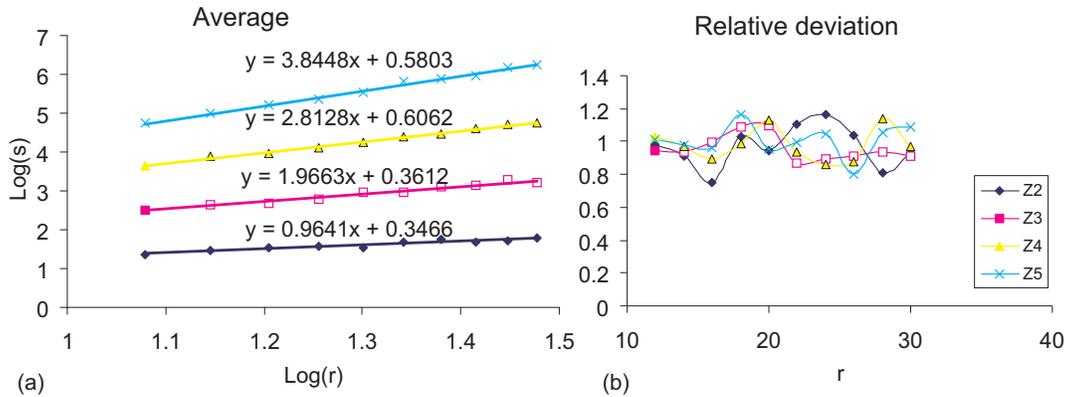}}
\caption{Number of SSA steps until first passage on an N-dimensional unit hypercube $Z_{N}$. (a) A $Log_{10}-Log_{10}$ plot of the average number of steps $s$ versus rate ratio $r$ (b) Relative deviation $ \frac{\surd\langle \delta s^{2} \rangle}{\langle s \rangle}$}
\label{ZNSSAsteps}
\end{figure}

\subsection{\label{MEnet}Master equation for bond networks} 
\subsubsection{Spectral analysis for breaking the $C_{N}$ network} 
We illustrate the Master Equation spectral method using the
cycle graph $C_{N}$ as an example. This is a graph of $N$ vertices and
$N$ edges, connected together in a loop such that exactly two edges
need to be removed to disconnect the graph (called a \emph{separation
pair}). The state space is $S = \{ \bm{0}\}
\bigcup_{\mu=1}^{N}\{\hat{\bm{\mu}}\}
\bigcup_{\mu=2}^{N}\bigcup_{\nu<\mu}\{\hat{\bm{\mu}} +
\hat{\bm{\nu}}\}$.  In this case, the subspace $V_{b} =
\bigcup_{\mu=2}^{N}\bigcup_{\nu<\mu}\{\hat{\bm{\mu}} +
\hat{\bm{\nu}}\}$ defines the absorbing boundary and the subspace
$V_{c}=S - V_{b}$ defines the space of \emph{transient} states. We
begin with the most general form for $M$, the projection of $W$ onto
the subspace $V_{c}$.

{\setlength\arraycolsep{1pt}\begin{equation}
M = \left(
\begin{array}{cccccc}
-\sum_{\mu}b_{\mu} & a_{1} & a_{2} & \ldots & a_{N} \\
 b_{1} & -(a_{1}+\sum_{\mu\neq 1}b_{\mu}) & 0 &  \ldots  & 0  \\
 b_{2} & 0 & -(a_{2}+\sum_{\mu\neq 2}b_{\mu}) &  0 &  \ldots\\
   \ldots    &   \ldots   &  0 &  \ldots    &  0  \\
   b_{N}    &   0    &   \ldots   &    0    & -(a_{N} +\sum_{\mu\neq N}b_{\mu})\\
\end{array}\right)
\label{CNMat}
\end{equation}}

In what follows, we assume that all the eigenvalues of $M$ are
negative (as they must be over the subset of \emph{transient states}
since $\sum_{n}M_{nm} \le 0$ ensures any positive probability density
decays to zero) and that the set of rates $\{a_{i}\}$ and $\{b_{j}\}$
are positive (ensured by property 1 of $W$). For economy of notation,
let us define $k_{n}= a_{n}+\sum_{m\neq n}b_{m}$. Also, in what
follows we assume that the bond indices have been labeled such that
$k_{1}\le k_{2}, \ldots k_{\alpha} \le k_{\alpha+1} \ldots \le
k_{N}$. In the case of a $C_{N}$ network, the $T_{f}(i)$ distribution
can be efficiently sampled due to certain properties of the eigenvalue
distribution and the form of the eigenvectors.  Since the sampling
technique for a general CTMM will be an extension of this special
case, it will be helpful to illustrate the method by investigating the
spectral properties of $C_{N}$. The next few results establish bounds
on the eigenvalues of $M$ as a special case of the interlacing
eigenvalue theorem~\cite{Horn}.

\begin{theorem} The $N+1$ eigenvalues $\{-\lambda_{0}>-\lambda_{1}> \ldots -\lambda_{N}\}$ of the matrix $M$ in Eq. \ref{CNMat} satisfy the following:
\begin{enumerate}
\item If $k_{i}=k_{i+1}$ then $-k_{i}$ is an eigenvalue of $M$. If $n$ such diagonal elements are identical then the eigenvalue is $(n-1)$-fold degenerate.
\item There is at least one eigenvalue of $M$ in the interval $\epsilon_{i} \equiv (-k_{i},-k_{i+1})$.
\end{enumerate}
\end{theorem}
\begin{proof} The \emph{eigenvalue} condition $Det|M - \lambda I|=0$ implies that the eigenvalues $\lambda$ are the zeroes of an $(N+1)^{th}$ order polynomial:
\begin{equation}
f(\lambda)= \left(\sum_{\mu}b_{\mu}+\lambda\right)\prod_{i=1}^{N}(k_{i} + \lambda) - \sum_{i=1}^{N}a_{i}b_{i}\prod_{j\neq i}(k_{j} +\lambda) =0
\end{equation} 
We establish bounds on the roots by calculating the sign of $f(\lambda)$ over the set of points $\{-k_{1}, \ldots, -k_{N}\}$.
\begin{enumerate}
\item Each term inside the summation sign in $f(\lambda)$ contains $n-1$ factors of $(k_{i} + \lambda)$. Hence  $-k_{i}$ is an $(n-1)$-fold degenerate eigenvalue. In what follows we assume that the remaining $k_{j}$ are all distinct.
\item The sign of the function $f(\lambda)$ at $\lambda=-k_{i}$ is $(-1)^{i}$. Hence $\epsilon_{i}$ encloses at least one root of $f(\lambda)$.\qedhere
\end{enumerate}
\end{proof}

The eigenvectors of $M$  $ \{|\psi_{\alpha}\rangle\}$ are mutually orthogonal for the set of non-degenerate eigenvalues.
 In the case of  non-degenerate eigenvalues ($-\lambda_{\alpha} \neq -k_{m}$), these eigenvectors are:
\begin{eqnarray}
\psi_{\alpha,n}&=& \langle n| \psi_{\alpha}\rangle = N_{\alpha}\frac{a_{n}}{(k_{n}-\lambda_{\alpha})}\\
\psi_{\alpha,0}&=& \langle 0| \psi_{\alpha}\rangle = N_{\alpha}
\end{eqnarray}
where $N_{\alpha}$ is a normalization constant.
For degenerate eigenvalues, an orthogonal basis can always be chosen using the \emph{Gram-Schmidt} procedure.  As will become apparent later, however, these eigenvalues do not contribute to the sampling in the case of a first-passage problem beginning with the unbroken loop $i=0$. 

\begin{theorem}\label{DisMix} The envelope curve $g(t)$ defined by our method is identical to the first passage density for a $C_{N}$ network.\end{theorem}
\begin{proof} 
Beginning with an unbroken state at $t=0$, the probability the model occupies a given state $n$ at time $t$ is given by:
\begin{equation}p_{n}(t)=\pi_{n} \sum_{\alpha=0}^{N}\psi_{\alpha,n}\psi_{\alpha,0}\exp[-\lambda_{\alpha}t] =\sum_{\alpha=0}^{N}c_{\alpha,n}\exp[-\lambda_{\alpha}t]\end{equation}
where $\psi_{\alpha,i}$ is an eigenvector of $\mathcal{M}$ with
eigenvalue $-\lambda_{\alpha}$. Note that only those
$\lambda_{\alpha}\neq k_{i}$ contribute, for otherwise
$\psi_{\alpha,0} = 0$. Assuming $\lambda_{\alpha}<\lambda_{\alpha+1}$,
the coefficients satisfy $c_{\alpha,n}< 0$ for $\alpha>n$. Since the
partial sum $S_{N,n}=\sum_{\alpha=0}^{N} c_{\alpha,n}=\sum_{\alpha=0}^{N}\pi_{n}\psi_{\alpha,n}\psi_{\alpha,0}=
\pi_{n}\langle n|0\rangle=0$ , all other partial sums satisfy $S_{\beta,n}=\sum_{\alpha=0}^{\beta} c_{\alpha,n}
\geq 0$.  These observations provide a means of
decomposing the probability density into the following discrete
mixture with positive coefficients:
\begin{equation}
p_{n}(t)= \sum_{\alpha=1}^{N}S_{\alpha,n}(\exp[-\lambda_{\alpha}t]-\exp[-\lambda_{\alpha +1}t])
\end{equation}
Since $b_{n} >0$, the combined rate of decay to any one of the broken states is given by:
\begin{equation}
\frac{dp^{B}}{dt}\equiv \sum_{(n,m)}\frac{dp_{(n,m)}}{dt}=\sum_{\alpha=0}^{N-1}S_{\alpha}f_{\alpha}(t)
\end{equation}
where,
\begin{equation}
 S_{\alpha}=\left(\frac{\lambda_{\alpha+1}-\lambda_{\alpha}}{\lambda_{\alpha}\lambda_{\alpha+1}}\right)\sum_{m} \sum_{n\neq m}\pi_{(n,m)}\left(b_{m}S_{\alpha,n}\right) >0\end{equation}
 and \begin{equation}f_{\alpha}(t)=\frac{\lambda_{\alpha}\lambda_{\alpha+1}}{\lambda_{\alpha+1}-\lambda_{\alpha}}(\exp[-\lambda_{\alpha}t] - \exp[-\lambda_{\alpha+1}t])\end{equation}\qedhere
\end{proof}

\subsection{\label{netmod}Simulation models used for bond networks}

Although our methods can in principle sample escape times from any
subnetwork of a CTMM graph, we have validated them here for the
specific case of breaking networks of bonds due to the importance of
this problem for self-assembly modeling. In rule-based models of
self-assembly, a simulation is initialized with a set of assembly
subunits, each with a complement of pre-specified binding sites.  As
the simulation progresses, the system evolves into a state with an
assembly of disjoint networks.  The binding interactions between two
disconnected pieces of the network usually occurs on a slower scale
than individual bond breaking reactions\cite{Zhang06}. For
bi-connected networks, however, the association rate within a
connected network is much larger than the bond breaking rate since
there is no entropy penalty in associating bonds between constituent
subunits. Such models allow for a natural partitioning of the state
space into subgraphs corresponding to the bi-connected components of
the entire network. The first set of experiments that we performed
were on such bi-connected networks. The simplest non-trivial example
of a bi-connected bond network is the graph generated by an $N$-cycle
($C_{N}$). More complicated networks of $N$ bonds can be viewed as
special cases of a truncated unit hypercube in $N$ dimensions. We
therefore carried out simulations for the network generated by $C_{N}$
as well as the full hypercube ($Z_{N}$).  Theorem \ref{SSAcon}
guarantees that the expected number of SSA steps for a $k$-connected
network of $d$ bonds is $P(d,k)\Omega(r^{k-1})$, where $P(d,k)$ is
some combinatorial function dependent on the topology of the network.

Each model is parameterized by a rate of bond formation, $a$, and a
rate of bond breaking, $b$.  These values were varied in different
simulations.  Each of the bonds had different binding/breaking rates
but the ratio was maintained at the same order of magnitude for each
simulation. Specifically, for a $d$ bond network $b_{\mu}= b(1.0 +
0.05\mu/d)$ and $a_{\mu}=a$.  These slight variations in rates from
bond to bond were used to avoid giving our methods an unfair advantage,
as they will generally be more efficient when the transition matrix has
degenerate eigenvalues.

\subsection{Experiments}

We conducted a series of simulations to determine the
performance of the SSA, Master Equation, and EMC methods for bond
network first-passage times.  All simulations were implemented in
Mathematica.  Run time simulations were executed on a Macintosh
machine with a 1.8GHz G5 processor and 512 MB RAM. For the EMC based
spectral method, we allowed each component of the state vector to
converge within a relative error of $\epsilon = 0.01$.  Each data
point reported was the average over 500 simulations except for run
time data, which were averaged over 100 simulations.

We first examined the efficiency of the Master Equation method by
assessing the number of rejection steps needed to sample each
first-passage time.  We carried out simulations for cycle graphs
($C_N$) varying the cycle length from 3 to 7 and the rate ratio $a/b$
from 1 to 20 in increments of 1.  These experiments were then repeated
for unit hypercubes ($Z_N$) with dimension varied from 2 to 5 and rate
ratio $a/b$ from 1 to 10 in increments of 1. For each condition, we
recorded the number of rejection steps required for each of 500
simulations and computed the mean and standard deviation across the
500 trials.

We next examined the number of steps required by the EMC method for
sampling times to network breakage.  We examined the same models as
those used to validate the Master Equation method: cycles of length 3 to 7
with rate ratios from 1 to 20 in increments of 1 and hypercubes of
dimension 2 to 5 with rate ratios from 1 to 10 in increments of 1.  We
similarly recorded the number of EMC steps required for each of
500 simulations and computed the mean and standard deviation across
the 500 trials.  We also computed the fraction of models that reached
the first passage time before relaxing to the slowest decay mode.

We next tested the total run time of each of the three methods on a
broader set of parameter ranges.  We evaluated run times for each
method for cycle networks of sizes 3 through 7.  We performed two sets
of evaluations for each.  The first set varied the rate ratio $a/b$
from 500 to 5000 in increments of 500 to provide a broad view of the
relative run times of the three methods.  These numbers span ranges of
values likely for protein assemblies.  For example, Zlotnick \emph{et
al.}\cite{Zlot99} have estimated a binding free energy of $\Delta G =
4.2$ kcal/mole for ODE based simulation of the kinetics of the
Hepatitis B virus, which yields $a/b=\exp{(\Delta G/RT)} \sim 1200$.
We then examined ratios of SSA to Master Equation and SSA to EMC run
times for each data point based on averages over 100 simulations per
parameter set.  In a second set of experiments, designed to give a
finer view of where each method is dominant in parameter space, we
varied the rate ratio $a/b$ from 30 to 300 in increments of 30.  We
then identified the most efficient of the three methods for each
point, again using averaged run times over 100 trials per data point.

We then performed analogous experiments for hypercube graphs in order
to test performance on networks with higher connectivity.  For each
graph $Z_{2}$ to $Z_{5}$, we carried out simulations for rate ratio
$a/b$ from 3 to 30 in increments of 3.  We were limited to small
ratios because the SSA method becomes prohibitively costly for
high-connectivity networks at higher ratios.  Each simulation was
repeated 100 times to yield average run times for each parameter set
and for each of the three methods.  For each parameter set, we
computed the ratio of run times for SSA versus Master Equation and
SSA versus EMC.  We further evaluated which of the three methods
produced the shortest average run time for each parameter set.

\subsection{Results}
We first present results on the efficiency of the rejection sampling
scheme for the Master Equation method.  The expected run time of
the method is proportional to the expected number of trials needed to
produce a successful sample.  A low number of steps is therefore
preferable, with a value of one being ideal.  Fig.~\ref{CNMEsteps}(a)
shows the rejection ratio for cycle graphs $C_3$ through $C_7$.  The
mean number of rejection steps is consistently below 1.5, as expected
from theorem \ref{DisMix} and \ref{RRatio}.  The number of rejection
steps drops with increasing rate ratio but increases with increasing
cycle length.  These results together establish the efficiency of the
method.  Fig.~\ref{CNMEsteps}(b) shows that the method is also robust,
with standard deviation consistently below 0.9 for the experiments
shown here.  The standard deviation also decreases with increasing
rate ratio but increases with cycle size.

Fig.~\ref{ZNMEsteps}(a) shows mean numbers of rejection steps for
hypercube graphs.  Since the envelope curve for hypercubes is not
exact, these experiments provide information about how well the method
performs for more general networks.  The hypercube graphs also yield
mean numbers of rejection steps consistently below 1.5. The number
of steps generally falls with increasing rate ratio. Fig.~\ref{ZNMEsteps}(b)
shows the method also to be robust for hypercube graphs, with standard
deviations consistently below 1.0 and following similar trends to the
means.

\begin{figure}
\centerline{\includegraphics[scale=0.7]{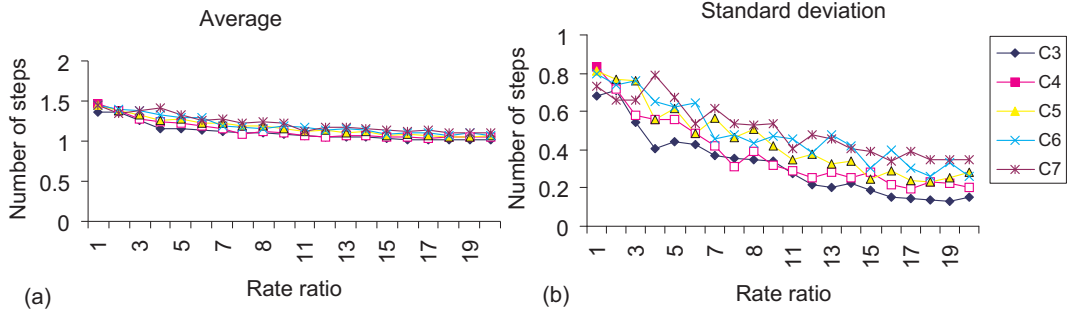}}
\caption{Number of Rejection steps for the Master Equation method until first passage for the network  generated by $C_{N}$ (a) Average number of steps $\langle s \rangle$ (b) Standard deviation $\surd\langle \delta s^{2} \rangle$}
\label{CNMEsteps}
\end{figure}

\begin{figure}
\centerline{\includegraphics[scale=0.7]{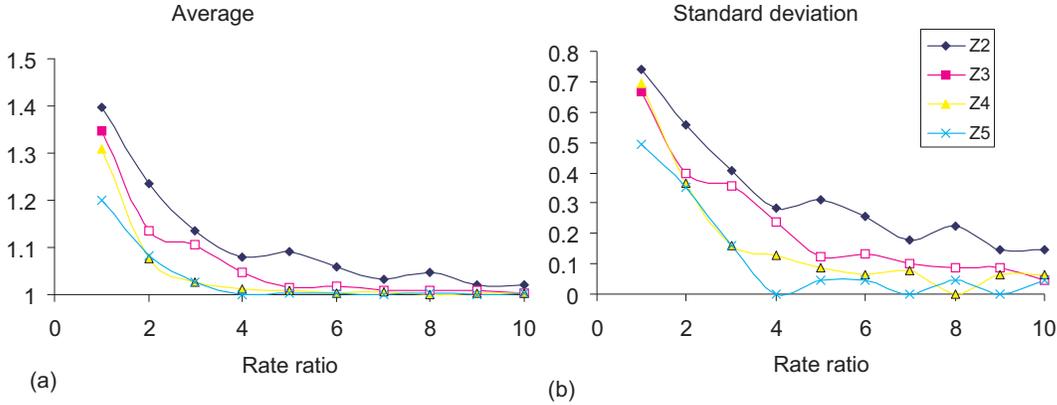}}
\caption{Number of Rejection steps for the Master Equation method until first passage for $Z_{N}$ (a) Average number of steps $\langle s \rangle$ (b) Standard deviation $\surd\langle \delta s^{2} \rangle$}
\label{ZNMEsteps}
\end{figure}

Next, we performed identical experiments to study the performance of
the EMC-based spectral method. Fig.~\ref{CNEMCsteps}(a) shows mean
numbers of EMC steps for cycle graphs.  The number of steps remains
consistently below 6.  The values rise sharply at the lowest rate
ratios, but quickly level off to approximately 4-5, depending on the
cycle length.  Figs.~\ref{CNEMCsteps} (b) and (c) provide the
explanation for this feature. For small rate ratio, multiple eigen
modes are responsible for the decay (see part (c)), which corresponds
to increasing EMC steps before first passage, similar to SSA. However,
as rate ratio increases further, relaxation time to the slowest eigen
mode becomes smaller than the average first passage time and the
method automatically samples breaking times according to the slowest
eigen mode (Fig \ref{CNEMCsteps}(c)). This feature is evident in part
(b) of the figures, which measure the standard deviation. At high rate
ratio the ``trajectory'' is almost deterministic, i.e., it always
takes the same number of steps to break the network. This happens
because the state almost always relaxes to the slowest eigen mode
before escaping the subgraph, hence giving a low value for $\sigma$ at high
rate ratio.

Fig.~\ref{ZNEMCsteps} shows comparable results for hypercube graphs.
Fig.~\ref{ZNEMCsteps}(a) shows that mean numbers of steps drop
substantially between ratios 1 and 2 but quickly level off to an
apparent constant for each graph.  The number of steps increases with
increasing hypercube dimension.  Figs.~\ref{ZNEMCsteps}(b) and (c)
again show that the method has high variability for low rate ratios,
where multiple eigen modes contribute significantly to the time
distribution and the method must behave similarly to the standard SSA.
At higher ratios, though, the slowest mode quickly dominates and the
number of steps required becomes highly reproducible.

\begin{figure}
\centerline{\includegraphics[scale=0.7000]{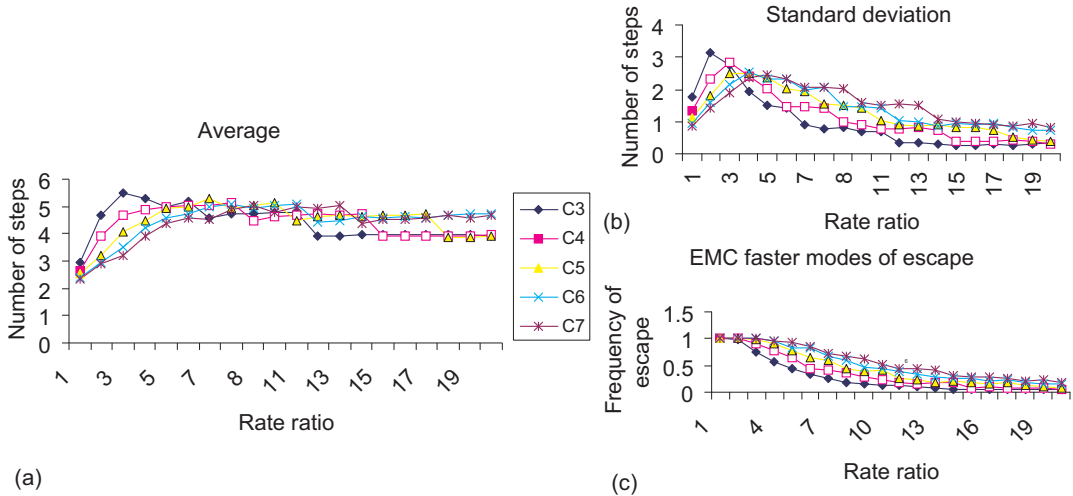}}
\caption{Number of EMC steps until first passage for the network  generated by $C_{N}$ (a) Average number of steps $\langle s \rangle$ (b) Standard deviation $\surd\langle \delta s^{2} \rangle$ (c) Fraction of times the trajectory escapes before relaxing to the slowest decay mode.}
\label{CNEMCsteps}
\end{figure}

\begin{figure}
\centerline{\includegraphics[scale=0.7000]{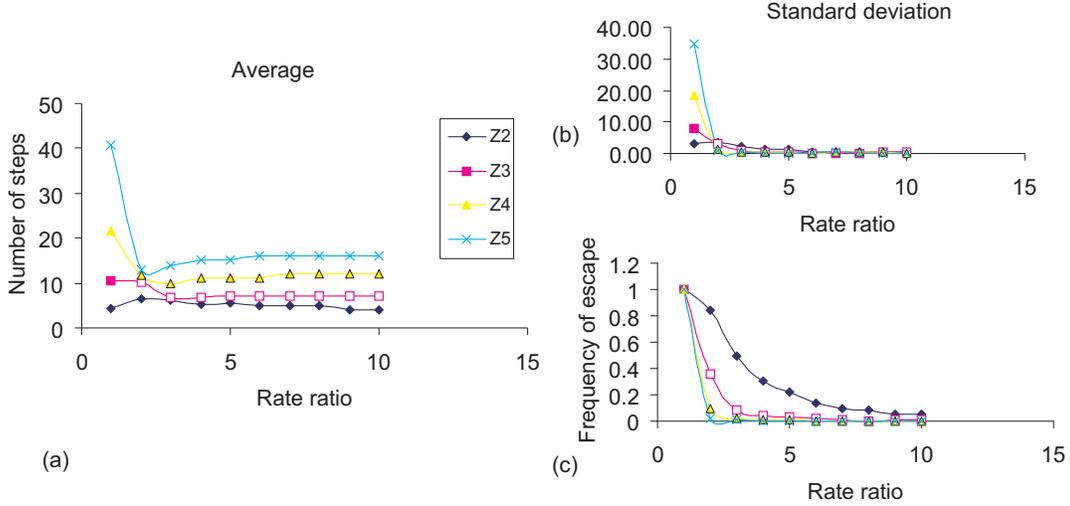}}
\caption{Number of EMC steps until first passage on $Z_{N}$ (a) Average number of steps $\langle s \rangle$ (b) Standard deviation $\surd\langle \delta s^{2} \rangle$(c) Fraction of times the trajectory escapes before relaxing to the slowest decay mode.}
\label{ZNEMCsteps}
\end{figure}

We next examined total run times of the three methods, beginning with
the cycle graphs $C_{3}$ to $C_{7}$.  Fig.~\ref{cyclefig} plots
results of the EMC and Master Equation methods relative to the basic
SSA. Fig.~\ref{cyclefig}(a) shows ratios of run times for standard SSA
to the Master Equation method.  The ratio grows rapidly with
increasing rate ratio, although it falls with increasing cycle size.
Fig.~\ref{cyclefig}(b) shows the comparison of SSA to the EMC method.
The SSA:EMC ratio likewise peaks for large rate ratios and small cycle
sizes.  The EMC method appears generally superior to the Master
Equation method, beginning to dominate at a lower rate ratio and
reaching a higher peak.  Fig.~\ref{cyclefig}(c) shows for a narrower
rate range where each of the three methods dominates.  The EMC method
is the fastest for most of the range examined, with the standard SSA
superior at the extreme of low ratios and large cycle sizes.

\begin{figure}
\includegraphics[scale=0.700]{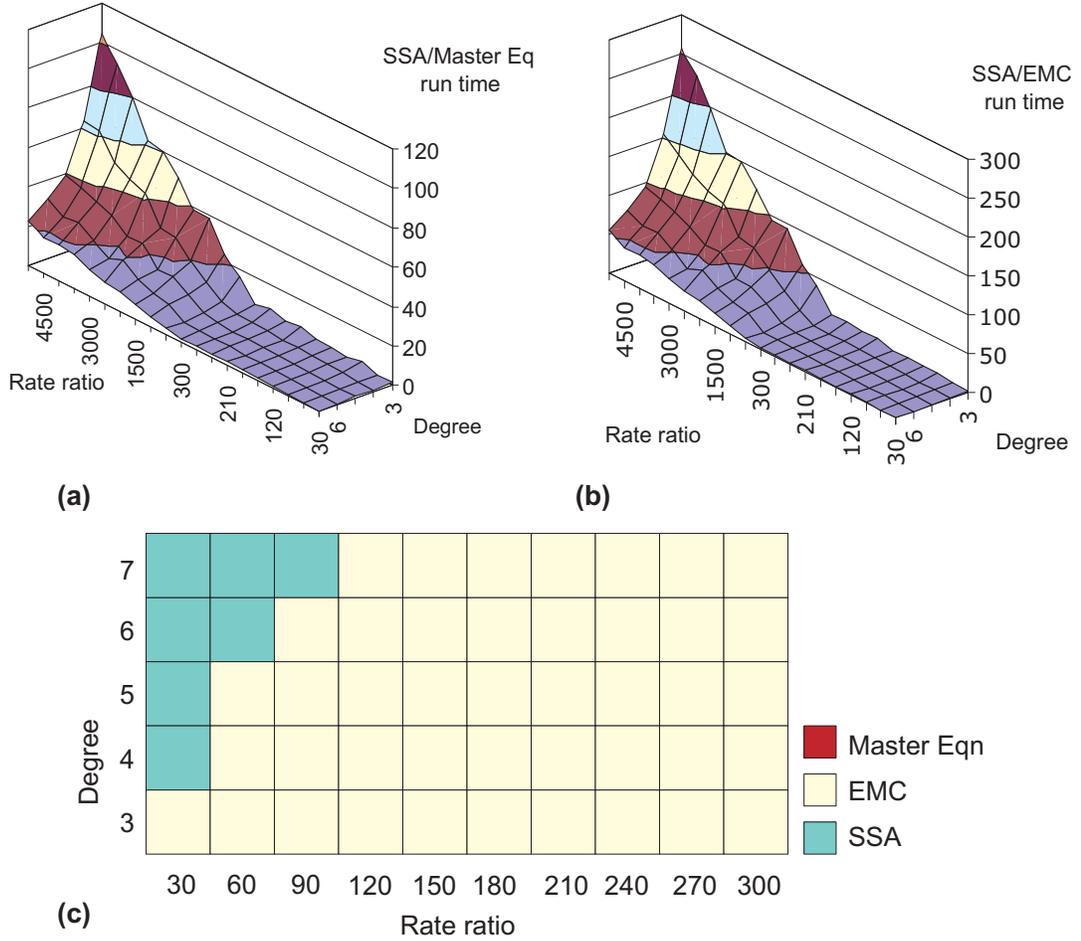}
\caption{Comparative run times for the network generated by $C_{N}$ (a) Ratio of SSA to Master Equation run times (b) Ratio of SSA to EMC run times (c) Region in 2D parameter space where each method is optimal }
\label{cyclefig}
\end{figure}
\begin{figure}
\centerline{\includegraphics[scale=0.700]{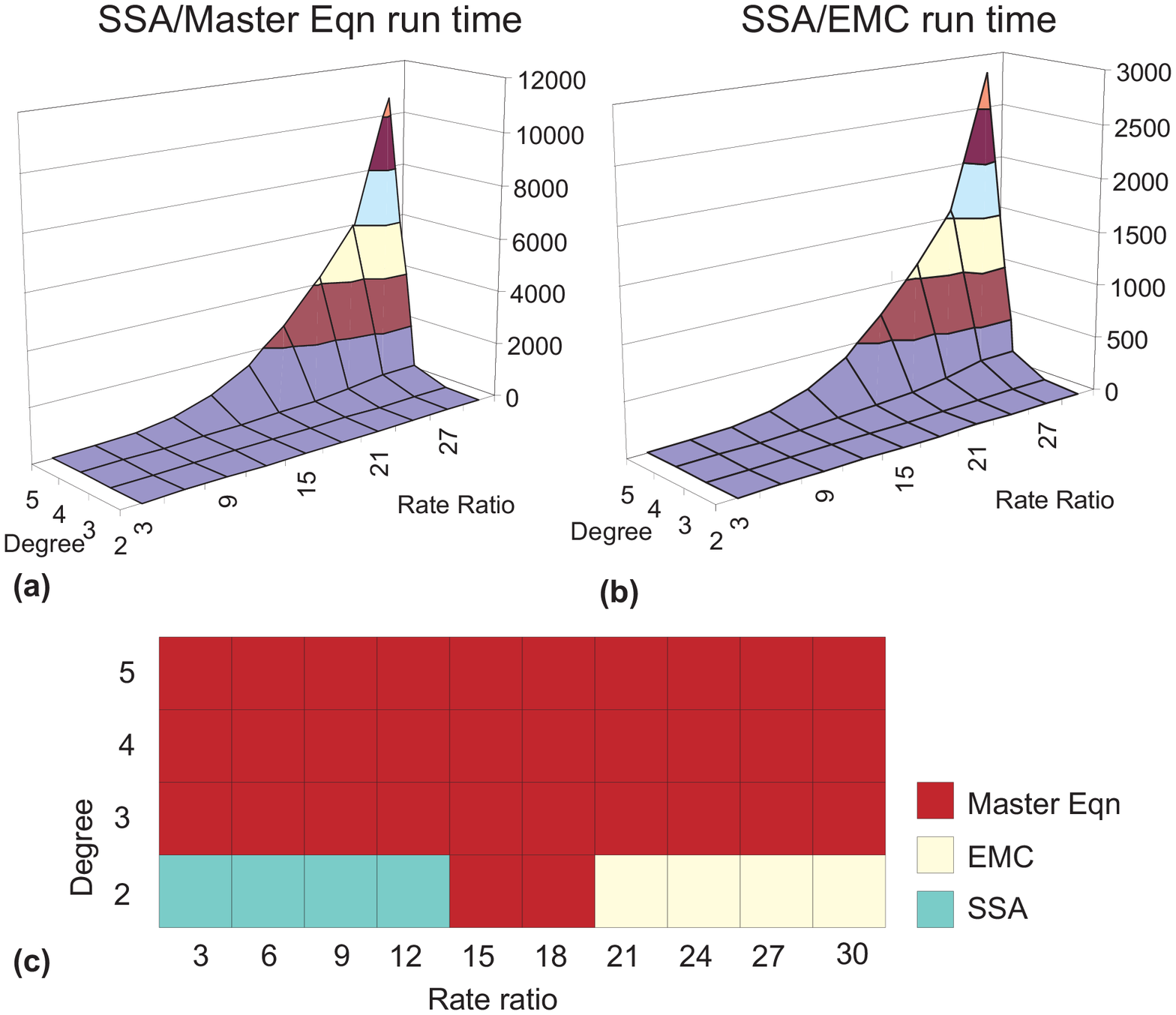}}
\caption{Comparative run times for first passage on $Z_{N}$ (a) Ratio of SSA to Master Equation run times (b) Ratio of SSA to EMC run times (c) Region in 2D parameter space where each method is optimal}
\label{cubefig}
\end{figure}

We then examined run times on the hypercube graphs $Z_2$ to $Z_5$.
Fig.~\ref{cubefig}(a) shows run time ratios for SSA versus the Master
Equation method and Fig.~\ref{cubefig}(b) for SSA versus the EMC
method.  Both spectral methods show sizable improvements over the
pure SSA method for larger rate ratios and higher hypercube
dimensions.  SSA appears much more sensitive to rate ratio as compared
to the spectral methods. Even for a rate ratio of 30, the spectral
methods were more than three orders of magnitude more efficient than
SSA for $Z_{5}$.  For hypercubes, unlike cycle graphs, the Master
Equation method appears generally more efficient than the EMC-based
method, even for small rate ratios.  Fig.\ref{cubefig}(c) shows where
each method is dominant.  The Master Equation method is dominant for
most of the parameter range examined, with the SSA method superior at
the limit of lowest degree and smallest rate ratios and the EMC
dominant for low degree and higher rate ratios.  This result is
expected from Fig. \ref{ZNSSAsteps}(a), since the average number of
steps seems to increase monotonically with the connectivity of the
graph for the EMC-based method. The efficiency of the Master Equation
method, on the other hand, depends primarily upon the size of the
complete graph, since matrix diagonalization is the eventual
efficiency bottleneck.

\section{\label{NuclLim}Application with Automated Discovery (AD): Nucleation-limited Assembly}
In this section, we apply the AD variants of the methods to a
different system type also motivated by self-assembly modeling.  The
rate of a self-assembly processes is often limited by the time
required to build the first stable multi-subunit complex, called a
nucleus, which then acts as a seed for assembly of the rest of a
larger structure.  Because partially formed nuclei are unstable,
considerable trial-and-error may be needed before one reaches
completion.  The time to complete a single nucleus can thus be orders
of magnitude longer than the inter-subunit binding rate.  These
nucleation-limited assembly systems are one example of the broader
class of stiff models for chemically reacting species.  The state
space of any such a system can be represented as a lattice
corresponding to the populations individual species.  These models are
similar to those treated in earlier studies of accelerated SSA methods
\cite{All05,Cao05}.  We apply one such model, representing the
formation of simple trimeric nuclei, to demonstrate and evaluate the
AD variants of the spectral methods.

\subsection{Integer lattice models}
The second model we consider is again an assembly of bond networks where monomers ($m$) with two identical binding sites combine to form dimers ($d$) and trimers ($t$). In order to show stiffness with respect to a single parameter, the trimers were assumed to be completely stable. If the total number of monomer subunits is $N$, the state space is the intersection of the plane $N_{m} + N_{d} + N_{t}= N$ with the positive octant of the 3 dimensional lattice formed by integer counts of the monomer ($N_{m}$), dimer ($N_{d}$) and trimer ($N_{t}$) populations. Let us represent each vertex of this graph by the pair $(N_t, N_d)$. The reaction propensities $\alpha_{N_t, N_d}^{N'_t,N'_d}$ to reach the vertex $(N'_t, N'_d)$ from $(N_t,N_d)$ are (to within an overall constant)
\begin{eqnarray}
\alpha_{N_t, N_d}^{N_t,N_d +1}&=& N_m(N_m-1)/v\\
\alpha_{N_t, N_d}^{N_t,N_d-1}&=& N_d \\
\alpha_{N_t, N_d}^{N_t+1,N_d-1}&=& N_m N_d/v \\
\alpha_{N_t, N_d}^{N_t-1,N_d+1}&=& 0
\end{eqnarray}
where  $1/v$ is an entropy penalty due to the finite volume of the system. We initialize the system at the state $(0,0)$ for a given monomer count $N$ and sample the first passage time until the trimer count reaches a given value. This system will show stiffness if the parameter $\rho\equiv N/v$ is small. For small $\rho$, which corresponds to low concentration and/or small binding energy, trimer formation will be much slower than dimer breaking/binding reactions.

\subsection{\label{latticeAD}Automated Discovery for integer lattice}
Efficient simulation over an integer lattice, where one pair of species react on a much faster timescale than the others requires a partitioning of the entire lattice into subgraphs with fixed trimer count (since trimer formation occurs on a much longer timescale than monomer-dimer reactions). These subgraphs are simple paths with vertex set $V(N_t)=\{(N_t,0),(N_t,1),\ldots (N_t, [(N-3N_t)/2])\}$, where $N_t$ represents the fixed trimer count and square brackets represent the largest integer smaller than the enclosed expression. Fig. \ref{Pathpcode} presents a procedure for implementing automated discovery on such graphs which works in optimal time, to within a small constant factor (the vertices are represented by dimer count for simplicity).
\begin{figure}
\includegraphics[scale=0.9]{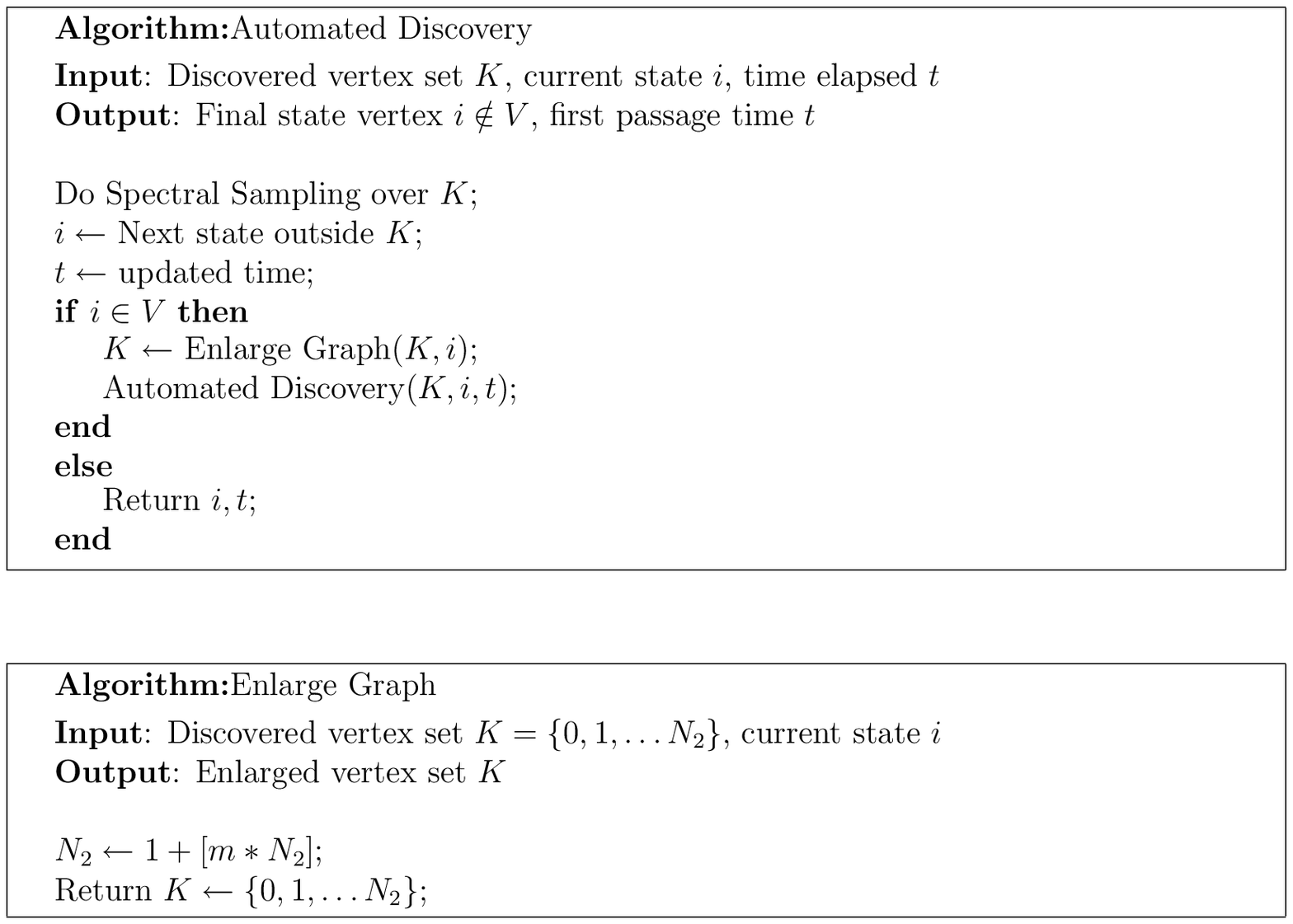}
\caption{Pseudocode for Automated Discovery for a simple path}\label{Pathpcode}
\end{figure}
At each step of automated discovery the method enlarges the graph by a factor proportional to its present length. The scale factor $m$ can be optimized for any given sampling algorithm to optimize run time. For example, if a given spectral sampling algorithm works in time $f(N) = N^\alpha$, where $N$ is the cardinality of the vertex set; total number of steps used in automated discovery of a graph sized $N_i$ can be bounded from above by the following quantity $S_{N_i}$,
\begin{eqnarray}
S_{N_i}&=&\sum_{n=1}^{1+\log_m{[N_i]}}m^{n\alpha} \leq \frac{N_i^{\alpha}m^{2\alpha} -1}{m^\alpha -1} \nonumber\\
&\leq& C N_{i}^{\alpha} \textrm{ for } m = 2^{1/(\alpha+1)}
\end{eqnarray}
where, $C= 2^{2\alpha/\alpha +1}/(2^{\alpha/\alpha+1} -1)$. In general
the Master Equation based method works in $O(N^3)$, however for simple
paths the Kolmogorov matrix is sparse and effective power is expected
to be more like $\alpha =2$. Since $S_{N_i}$ is more sensitive to
deviations for smaller values of $m$, we chose $m=1.3 >2^{1/3}$ in the
experiments reported here.

The method reported here can in principle be generalized for arbitrary lattice graphs in $d$ dimensions. The size of the discovered graph in such cases would overestimate the actual trapped graph by a factor of $m^d$, for a scaling factor $m$. For small dimensions, this may still be more efficient than the method discussed in section \ref{AD} which exactly samples the trapped subgraph.

\subsection{Experiments}
We performed two sets of experiments to compare the performance of
spectral methods with SSA. The first set of experiments compared the
Master Equation method implemented in conjunction with Automated
Discovery for the trimer model with SSA. Each experiment compared the
ratio of run times for sampling first passage times to reach a trimer
count $N_t=100$, starting from an initial monomer count $N$. The state
space was partitioned into subgraphs corresponding to fixed trimer
counts and AD was used to identify the trapped
regions for spectral decomposition. We then performed a total of 50
comparative run time simulations varying $N$ from 1000 to 9000 in
steps of 2000 and varying $\rho$ from $10^{-5}$ to $1.9\times 10^{-4}$
in steps of $2.0\times 10^{-5}$. All run times were averaged over 50
samples. The scale factor for AD was set at 1.3.  The
second set of experiments compared the run time ratio for the EMC
based method and SSA for first passage time to reach a trimer count
$N_t=100$, starting from an initial monomer count $N$.  The state
space was again partitioned into subgraphs of fixed trimer
counts. AD was not required for these simulations
since the method automatically selects the trapped region of the
subgraph according to the evolving probability distribution. We then
performed a total of 25 comparative run time simulations varying $N$
from 1000 to 9000 in steps of 2000 and varying $\rho$ from $10^{-4}$
to $0.9\times 10^{-3}$ in steps of $2.0\times 10^{-4}$. All run times
were averaged over 50 samples. For the spectral method, each component
of the slowest eigenvector was allowed to relax to within a relative
error of 0.01 and an absolute error of $1.0\times 10^{-6}$.

\subsection{Results}
\begin{figure}
\centerline{\includegraphics[scale=0.900]{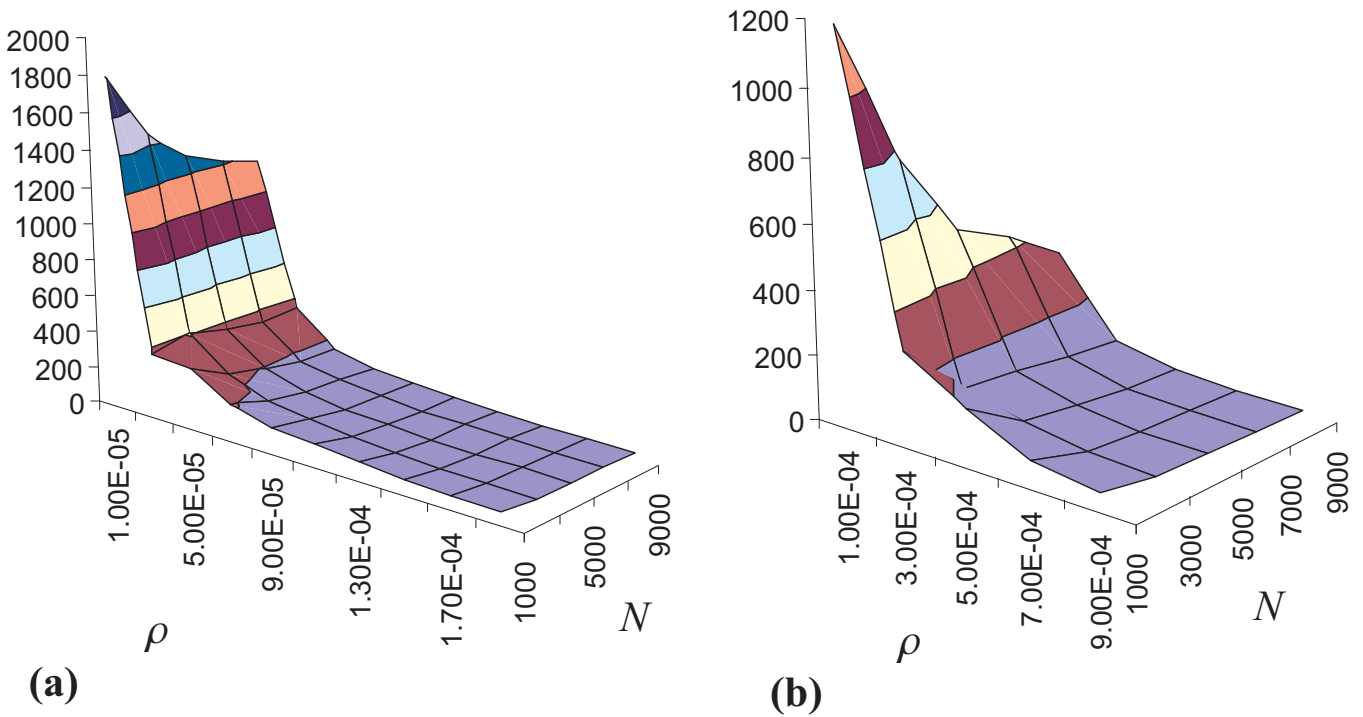}}
\caption{(a) Comparative run times for first passage for 10 trimer counts. Ratio of SSA to ME-AD run times  (b) Comparative run times for first passage for 100 trimer counts. Ratio of SSA to EMC-AD run times}
\label{latticefig}
\end{figure}

We first present results for the Master Equation method run times as
we vary the stiffness parameter $\rho$. Fig \ref{latticefig}(a) shows
the behavior for 5 different initial monomer counts $N$. Small $\rho$
values correspond to a stiff model, since the average dimer count
varies approximately as $N_d \sim \rho N_m$ and the ratio of the rate
of dimer formation to trimer formation varies as $\sim N_m/N_d$.
Fig.~\ref{latticefig}(a) demonstrates the efficiency of the Master
Equation method.  The method shows large gains in the domain of small
$\rho$ and small $N$, with relative performance dropping rapidly with
increasing $\rho$ and more slowly with increasing $N$.  Next we
present results for the ratio of SSA/EMC method run times as we vary
$\rho$ and $N$. Fig.~\ref{latticefig}(b) shows the behavior for 5
different initial monomer counts $N$.  The EMC method is effective at
substantially larger values of $\rho$ than is the Master Equation
method.  As is to be expected, the spectral methods do not scale as
well as the usual SSA method with increasing $N$.  Even for relatively
large networks, though, the performance gain obtained by spectral
sampling is appreciable. The reason for this is that for most cases
the slowest eigen mode is reasonably well approximated by a vector
populating only a small fraction of the subgraph vertices. As a result
we can look at the EMC method as a generalization of other accelerated
sampling schemes which only use one vertex, the mean value of the
slowest decay mode, as in the PEA based methods.

\section{\label{discussion}Discussion}

We have investigated the problem of efficiently simulating stochastic
reaction models and introduced two methods for accelerating sampling
on problems characterized by multiple time scales.  Both methods are
based on spectral analysis of CTMMs equivalent to the SSA model.  We
have applied these methods in the present work to two special cases of
these models that are important to simulations of molecular
self-assembly: sampling times to break multiply-connected bond
networks and simulating growth in nucleation-limited assembly systems.
Collectively, these two applications demonstrate the use of the
proposed spectral methods on small CTMM graphs known {\em a priori}
and on automatically discovered subgraphs of large CTMMs.  We have
shown theoretically and empirically that the new methods are
substantially more robust to variations in the ratios of reaction
rates than is the basic SSA method for these problems.

While we have applied these methods here to models used in
self-assembly simulations, the basic methods can be expected to have
much broader application.  Both methods can be applied to sample first
passage times for any arbitrary subset of states of any SSA CTMM
graph.  Both can also be applied to sample escape times from any
subgraph of such a graph, using automated discovery to identify
``trapped'' regions of the CTMM graph.  The latter distinction is
important because CTMM graphs for complicated biological systems are
generally far too large to represent explicitly.  These spectral
methods might be extended to incorporate ``on the fly'' graph
construction techniques, like those used by rule-based methods widely
used for SSA simulations~\cite{Moleculizer,BioNetGen}. The EMC method,
especially, would seem to be a candidate for such an extension. For
example, if at each iteration, instead of adding all the possible next
neighbors to the system state, we add only a subset of them depending
upon their transition probabilities then we will get a natural,
\emph{non-local} generalization of the SSA.  Such an approach could
provide a precise and general method for pruning full SSA graphs to
achieve more efficient pathway sampling in extremely large state
spaces.

\begin{acknowledgments}  This work was supported in part by a U.S. National Science Foundation award \# 0346981.
\end{acknowledgments}

\appendix
\section{}
In this appendix, we prove Thm.~\ref{SSAcon}, which helps us establish
the relative efficiency of the spectral methods to the standard SSA.
Before we prove the theorem, we need to establish some preliminary
results.  Let $r \equiv Min(a_{\mu}/b_{\nu}|\mu ,\nu \in\{1,\ldots
d\})$. To construct the transition matrix $Q$, SSA identifies the
negative of the diagonal element of the Kolmogorov matrix $ -
W_{\bm{n},\bm{n}}= \left(\sum_{\beta} a_{\beta}n_{\beta} + (1-
n_{\beta})b_{\beta}\right)$ as the inverse of the mean waiting time at
each SSA step and the matrix $ L_{\bm{m},\bm{n}} =
-W_{\bm{m},\bm{n}}/W_{\bm{n},\bm{n}}$ as the graph
Laplacian. Therefore, $Q_{\bm{m},\bm{n}} = \delta_{\bm{m},\bm{n}} -
L_{\bm{m},\bm{n}}$.  Since SSA simulates a periodic Markov chain $Q$,
the graph is bipartite and the two step chain $Q^{2}$ is reducible
into $Q_{even}^{2} \oplus Q_{odd}^{2}$. Here, $Q_{even}^{2}$ is the
projection of $Q^{2}$ over the subspace of states with an even number
of bonds broken $V_{even}$ and $Q_{odd}^{2}$ is the projection over
$V_{odd}= S - V_{even}$. Since both $Q_{even}^{2}$ and $Q_{odd}^{2}$
are irreducible and aperiodic, the ergodic theorem applies to each one
separately and if $|\Pi\rangle = \sum_{i\in
S}\pi_{\bm{i}}|\bm{i}\rangle$ is the eigenvector of $Q$ with
eigenvalue 1, the vectors $|\Pi_{e}\rangle = \sum_{\bm{i}\in
V_{even}}\pi_{\bm{i}}|\bm{i}\rangle$ and $|\Pi_{o}\rangle =
\sum_{\bm{i}\in V_{odd}}\pi_{\bm{i}}|\bm{i}\rangle$ are the
equilibrium distributions for $Q_{even}^{2}$ and $Q_{odd}^{2}$,
respectively, up to a normalization constant. To bound the mean
hitting time $T_{\bm{b}\bm{0}}$, from $\bm{0}$ to the set $V_{b}$, we
first apply the common technique of constructing another graph with
vertex set $\bar{V}=V_{c} \bigcup \{\bm{b}\}$, where all vertices in
$V_{b}$ are truncated to a single vertex $\bm{b}$ and
$V_{c}=S-V_{b}$. The edge weights for edges from $\bm{i}\in V_{c}$ to
$\bm{b}$ are chosen as $Q_{\bm{b},\bm{i}}= \sum_{\bm{j}\in
V_{b}}Q_{\bm{j},\bm{i}}$, which will leave $T_{\bm{b}\bm{0}}$
unchanged from that of the original graph. We must further specify the
edge weights from $\bm{b}$ to any states with $k-1$ broken bonds. In
order to ensure that the Markov chain still obeys detailed balance, we
require that $Q_{\bm{i},\bm{b}}/Q_{\bm{j},\bm{b}}=
(Q_{\bm{b},\bm{i}}*\pi_{\bm{i}})/(Q_{\bm{b},\bm{j}}*\pi_{\bm{j}})$ and
$\sum_{\bm{i}\neq\bm{b}}Q_{\bm{i},\bm{b}} = 1$.  The resulting
modified graph will then have the same hitting time $T_{\bm{b}\bm{0}}$
as the original graph.
 
We next need three auxiliary results about properties of the resulting graph in order to prove our main theorem.
\begin{lemma}\label{Ctime}For an ergodic Markov chain, the cover time $C_{\bm{i}\bm{j}}\equiv T_{\bm{i}\bm{j}} + T_{\bm{j}\bm{i}}$ between any two states $\bm{i}$ and $\bm{j}$ satisfies~\cite{Aldous}
\begin{equation}
E[C_{\bm{i}\bm{j}}]=E[T_{\bm{i}\bm{j}}] + E[T_{\bm{j}\bm{i}}] = 1/(\pi_{\bm{j}}Pr[T_{\bm{j}\bm{j}}>T_{\bm{i}\bm{j}}])
\end{equation}
\end{lemma}

\begin{lemma}\label{Qbounds} The transition matrix $Q$ satisfies the following conditions:
\begin{enumerate}
\item If $\bm{i}$ and $\bm{j}=\bm{i} + \hat{\bm{\mu}}$ are two neighboring states with $n$ and $n+1$ bonds broken, respectively, then $Q_{\bm{j},\bm{i}}\leq (n*r)^{-1}$ for any $n>0$.
\item For any initial state $\bm{i}$ containing $n$ broken bonds, the $n$-step transition probability to $\bm{0}$ is bounded from below by $Q_{\bm{0},\bm{i}}^{n}\geq (1+d/r)^{-n}$. 
\item Let $T$ be any stopping time for the transition matrix $Q$ with expectation value $E[T]=\sum_{n=1} n Pr[T=n]=\sum_{n=0}Pr[T>n]$. For any integer $l>1$ consider the expectation value of $T$ for the $l$-step transition matrix $Q^{l}$ defined as $E^{(l)}[T] = \sum_{n} n Pr[(n-1)*l< T \leq n*l] =\sum_{n=0} Pr[T > n*l]$. Then, $l(E^{(l)}[T]-1)\leq E[T] \leq l E^{(l)}[T]$.
\end{enumerate}
\end{lemma}
\begin{proof}

\begin{enumerate}
\item  The transition probability corresponding to the matrix element connecting $\bm{i}$ to $\bm{j}$ is : 
\begin{equation}
Q_{\bm{j},\bm{i}}=\frac{b_{\mu}}{(\sum_{\beta} a_{\beta}i_{\beta} + (1- i_{\beta})b_{\beta})}\nonumber \leq
\left\{\begin{array}{cc}
(n*r)^{-1} & \textrm{ $\forall$ }\bm{i} \neq \bm{0} \nonumber \\
1 & \textrm{ if } \bm{i}  = \bm{0}\\
\end{array}\right.
\end{equation}
\item First consider any state $\hat{\bm{\mu}}$ with one bond broken: 
\begin{equation}
Q_{\bm{0},\hat{\bm{\mu}}}=\frac{a_{\mu}}{a_{\mu} +\sum_{\nu\neq\mu}b_{\nu}}\geq(1+d/r)^{-1}
\end{equation}
Assume $Q_{\bm{0},\sum_{i=1}^{n}\hat{\bm{\mu}}_{i}}^{n} \geq (1+d/r)^{-n}$ for all $n$ broken bond states $\sum_{i= 1}^{n}\hat{\bm{\mu}}_{i}$, then 
\begin{eqnarray}
Q_{\bm{0},\sum_{i=1}^{n+1}\hat{\bm{\mu}}_{i}}^{n+1}&=& \sum_{\nu\in\{\mu_{1},\ldots,\mu_{n+1}\}}Q_{\bm{0},\sum_{\mu_{i}\neq\nu}\hat{\bm{\mu}}_{i}}^{n}
Q_{\sum_{\mu_{i}\neq\nu}\hat{\bm{\mu}}_{i},\sum_{i=1}^{n+1}\hat{\bm{\mu}}_{i}}\nonumber \\
&\geq &(1+d/r)^{-n}\sum_{\nu\in\{\mu_{1},\ldots,\mu_{n+1}\}}\frac{a_{\nu}}{\sum_{i=1}^{n+1}a_{\mu_{i}} + \sum_{\eta\neq\{\mu_{i}\}}b_{\eta}}\nonumber \\
&\geq &(1+d/r)^{-n-1}
\end{eqnarray}
Since $Q_{\bm{0},\hat{\bm{\mu}}}>(1+d/r)^{-1}$, the assertion holds for all $n>1$ by induction.
\item We can prove the upper bound as follows:
\begin{eqnarray}
E[T]&=& \sum_{n=1} \left( \sum_{m=1}^{l} ((n-1)l +m)Pr[T=((n-1)l +m)]\right)\nonumber \\
&\leq& \sum_{n=1} n*l\left(\sum_{m=1}^{l}Pr[T=((n-1)l +m)]\right)\nonumber \\
&\leq& l\sum_{n=1} n Pr[ (n-1)l< T \leq n*l] \nonumber \\
&\leq& l E^{(l)}[T]\nonumber \\
\end{eqnarray}
We can similarly prove the lower bound:
\begin{eqnarray}
E[T]&=& \sum_{n=1} \left( \sum_{m=1}^{l} ((n-1)l +m)Pr[T=((n-1)l +m)]\right)\nonumber \\
&\geq&\sum_{n=1} (n-1)*l\left(\sum_{m=1}^{l}Pr[T=((n-1)l +m)]\right) \nonumber \\
 &\geq& l\sum_{n=1} n Pr[ (n-1)l< T \leq n*l] - l \nonumber \\
&\geq& l(E^{(l)}[T]-1)
\end{eqnarray}
\end{enumerate}
\end{proof}

\begin{lemma}\label{Atime}The expected hitting time from the vertex $\bm{b}$ to $\bm{0}$ is bounded  by 
\begin{equation}
k\leq E[T_{\bm{0}\bm{b}}]\leq k(1+d/r)^{k}
\end{equation}
\end{lemma}
\begin{proof}The lower bound is trivial since at least $k$ bonds must be repaired before any disconnected state can reach $\bm{0}$. Consider the $n*k$ step probability for transition from $\bm{b}$ to $\bm{0}$. Let $\tilde{Q}$ be the transition matrix restricted to the set $\tilde{V}=\bar{V}-\{\bm{0}\}$, i.e., 
\begin{equation}
\tilde{Q}_{\bm{i},\bm{j}}=Q_{\bm{i},\bm{j}}\left(1-\delta_{\bm{0}\bm{j}}-\delta_{\bm{i}\bm{0}} + \delta_{\bm{0}\bm{j}}\delta_{\bm{i}\bm{0}}\right)
\end{equation}
The probability of a trajectory starting at $\bm{i}\in \tilde{V}$ reaching $\bm{0}$ in $k$ steps or less is given by 
\begin{eqnarray}
Pr[T_{\bm{0}\bm{i}}\leq k]&=& 1-\sum_{\bm{j}\in\tilde{V}}\tilde{Q}_{\bm{j},\bm{i}}^{k}=\sum_{n=1}^{k}\left(\sum_{l\in\tilde{V}}
Q_{\bm{0},\bm{l}}\tilde{Q}_{\bm{l},\bm{i}}^{n-1}\right)\nonumber \\
&\geq& Q_{\bm{0},\bm{b}}^{k}>(1+d/r)^{-k}
\end{eqnarray}
where we have used Lemma \ref{Qbounds} part (2). Let us define $p= (1+d/r)^{-k}$. In terms of $p$, the previous inequality and Lemma \ref{Qbounds} part (3) imply
\begin{eqnarray}
Pr[T_{\bm{0}\bm{b}}> n*k]&=& 1- Pr[T_{\bm{0}\bm{b}}\leq n*k]\leq(1-p)^{n}\nonumber \\
\Rightarrow E^{(k)}[T_{\bm{0}\bm{b}}] &\leq& \sum_{n=1}^{\infty}(1-p)^{n}=1/p\nonumber \\
E^{(k)}[T_{\bm{0}\bm{b}}] &\leq& (1+d/r)^{k} \nonumber \\
\Rightarrow E[T_{\bm{0}\bm{b}}] &\leq& k(1+d/r)^{k}
\end{eqnarray}
\end{proof}
An immediate consequence of the previous lemma is that for $k$ even
\begin{equation}
k/2 \leq E^{(2)}[T_{\bm{0}\bm{b}}] \leq (k/2) (1+d/r)^{k} + 1
\end{equation}
Similarly, if $k$ is odd, using the fact that $Pr[T_{\bm{0}\bm{b}}<T_{\hat{\bm{\mu}}\bm{b}}]=0$ we get
\begin{equation}
\frac{k-1}{2}\leq E^{(2)}[T_{\hat{\bm{\mu}}\bm{b}}] \leq E[T_{\bm{0}\bm{b}}]/2 + 1 \leq \frac{k}{2}(1+d/r)^{k} + 1
\end{equation}
 Let us define the equilibrium probability for $\bm{b}$ as $\tilde{\pi}_{\bm{b}}$, then
\begin{eqnarray}
\tilde{\pi}_{\bm{b}}&=& \frac{\pi_{\bm{b}}}{\sum_{\bm{i}\in V_{even}} \pi_{\bm{i}}} \textrm{ if {\it k} is even}\nonumber \\
&=& \frac{\pi_{\bm{b}}}{\sum_{\bm{i}\in V_{odd}} \pi_{\bm{i}}} \textrm{  if {\it k} is odd}
\end{eqnarray}
We can finally compute upper ($U$) and lower ($L$) bounds on the hitting time $T_{\bm{b}\bm{0}}$ which are asymptotically equivalent in the limit $r\rightarrow\infty$. The following theorem implies that $\Delta(r)\equiv U(r)-L(r)$ is monotonically decreasing in $r$ and $\lim_{r\rightarrow\infty} \Delta(r)/L(r) =0$~\cite{BenderOrszag}.

\begin{theorem}The expected number of SSA steps before first passage on a $k$-connected graph is bounded within
\begin{equation}
2\frac{ (1-d/r)^{-1} - \left((k-1)/2\right)\tilde{\pi}_{\bm{b}}}{\tilde{\pi}_{\bm{b}}}\geq E[T_{\bm{b}\bm{0}}] \geq  2\frac{1 - (k/2(1 + d/r)^{k}+2)\tilde{\pi}_{\bm{b}}}{\tilde{\pi}_{\bm{b}}}
\end{equation}
\end{theorem}
\begin{proof} In order to apply lemma \ref{Ctime} to bound the hitting time we need to look at graphs with $k$ odd or even separately. If $k$ is even we can apply lemma \ref{Ctime} directly to $C_{\bm{0}\bm{b}}$ for the 2-step chain $Q^{2}_{even}$. However, if $k$ is odd, we need to consider the cover time between $\bm{b}$ and each state $\hat{\bm{\mu}}$ with exactly one broken bond. Then, using the fact $Q|\bm{0}\rangle =(1/\sum_{\nu}b_{\nu}) \sum_{\mu}b_{\mu}|\hat{\bm{\mu}}\rangle$, we get:
\begin{eqnarray}\label{kodd}
E[T_{\bm{b}\bm{0}}] &=& 1 + \frac{1}{\sum_{\nu}b_{\nu}}\sum_{\mu}b_{\mu}E[T_{\bm{b}\hat{\bm{\mu}}}]
\end{eqnarray}

Since $Pr[T_{\bm{b}\bm{b}}>T_{\bm{0}\bm{b}}] = \sum_{n>m} Pr[T_{\bm{b}\bm{b}}=n]Pr[T_{\bm{0}\bm{b}}=m]$, for the $k$-step chain discussed in lemma \ref{Atime} we get:
\begin{eqnarray}
Pr[T_{\bm{b}\bm{b}}<T_{\bm{0}\bm{b}}] &\leq& \sum_{n=1}^\infty \left(\frac{d}{r(1+d/r)}\right)^{n} = d/r\nonumber \\
\Rightarrow Pr[T_{\bm{b}\bm{b}}>T_{\bm{0}\bm{b}}] &\geq& 1 - d/r
\end{eqnarray}
Also, $Pr[T_{\bm{b}\bm{b}}>T_{\hat{\bm{\mu}}\bm{b}}] \geq Pr[T_{\bm{b}\bm{b}}>T_{\bm{0}\bm{b}}] \geq 1 - d/r$.
 Suppose $k$ is even.  Then we can estimate the cover time $C_{\bm{0}\bm{b}} = T_{\bm{b}\bm{0}} + T_{\bm{0}\bm{b}}$ using lemma \ref{Ctime}.
\begin{eqnarray}
E[T_{\bm{b}\bm{0}}]&\geq& 2*E^{(2)}[T_{\bm{b}\bm{0}}]-2 =2*\left( \frac{1}{\tilde{\pi}_{\bm{b}}Pr[T_{\bm{b}\bm{b}}>T_{\bm{0}\bm{b}}]} - E^{(2)}[T_{\bm{0}\bm{b}}] -1\right) \nonumber \\
&\leq& 2*E^{(2)}[T_{\bm{b}\bm{0}}] =2*\left( \frac{1}{\tilde{\pi}_{\bm{b}}Pr[T_{\bm{b}\bm{b}}>T_{\bm{0}\bm{b}}]} - E^{(2)}[T_{\bm{0}\bm{b}}] \right)
\end{eqnarray}
An analogous argument for odd $k$ on using Eq. \ref{kodd} gives,
\begin{eqnarray}
E[T_{\bm{b}\bm{0}}]&\geq& 1 + \frac{1}{\sum_{\nu}b_{\nu}}\sum_{\mu}b_{\mu}\left(\frac{2}{\tilde{\pi}_{\bm{b}}Pr[T_{\bm{b}\bm{b}}>T_{\hat{\bm{\mu}}\bm{b}}]}
-2*E^{(2)}[T_{\hat{\bm{\mu}}\bm{b}}]-2\right)\nonumber \\
&\leq& 1 + \frac{1}{\sum_{\nu}b_{\nu}}\sum_{\mu}b_{\mu}\left(\frac{2}{\tilde{\pi}_{\bm{b}}Pr[T_{\bm{b}\bm{b}}>T_{\hat{\bm{\mu}}\bm{b}}]}
-2*E^{(2)}[T_{\hat{\bm{\mu}}\bm{b}}]\right)
\end{eqnarray}
Finally, using lemma \ref{Ctime} and \ref{Atime} we get for all $k$,
\begin{eqnarray}
E[T_{\bm{b}\bm{0}}] &\geq & 2\frac{1 - \left(k/2(1 + d/r)^{k}+2\right)\tilde{\pi}_{\bm{b}}}{\tilde{\pi}_{\bm{b}}}\nonumber \\
&\leq & 2\frac{ (1-d/r)^{-1} - \left((k-1)/2 \right)\tilde{\pi}_{\bm{b}}}{\tilde{\pi}_{\bm{b}}}
\end{eqnarray}

\end{proof}

As a corollary to the preceding theorem we get the result stated in section ~\ref{BondNet}.

\noindent
{\bf Theorem~\ref{SSAcon}} The expected number of SSA steps required to break a $k$-connected network with $k>1$ and $r>1$ is $\Omega(r^{k-1})$.

\begin{proof}
 Let $\bm{i}$ and $\bm{j}=\bm{i} + \hat{\bm{\mu}} + \hat{\bm{\nu}}$ be
 two graphs with $c$ and $c+2$ bonds broken respectively. Since we are
 interested in computing the invariant distribution for the
 irreducible components $Q_{even}^{2}$ and $Q_{odd}^{2}$, we first
 compute each matrix element connecting $\bm{i}$ to $\bm{j}$:
\begin{eqnarray}
Q^{2}_{\bm{j},\bm{i}}&=& \sum_{p=\mu,\nu}Q_{\bm{j},\bm{i}+\hat{\bm{p}}}Q_{\bm{i}+\hat{\bm{p}},\bm{i}}\nonumber \\
&=& \sum_{p=\mu,\nu}\left(\frac{b_{\mu}b_{\nu}}{(\sum_{\alpha} a_{\alpha}(i_{\alpha}+\delta_{p\alpha}) + (1- i_{\alpha}-\delta_{p\alpha})b_{\alpha})(\sum_{\beta} a_{\beta}i_{\beta} + (1- i_{\beta})b_{\beta})}\right)\nonumber \\
&=& \frac{b_{\mu}b_{\nu}}{W_{\bm{i},\bm{i}}}\left(\frac{1}{W_{\bm{i}+\hat{\bm{\mu}},\bm{i}+\hat{\bm{\mu}}}} + \frac{1}{W_{\bm{i}+\hat{\bm{\nu}},\bm{i}+\hat{\bm{\nu}}}}\right)
\end{eqnarray}
similarly,
\begin{eqnarray}
Q^{2}_{\bm{i},\bm{j}}&=& \sum_{p=\mu,\nu}Q_{\bm{i},\bm{j}-\hat{\bm{p}}}Q_{\bm{j}-\hat{\bm{p}},\bm{j}}\nonumber \\
&=& \sum_{p=\mu,\nu}\left(\frac{a_{\mu}a_{\nu}}{(\sum_{\alpha} a_{\alpha}(j_{\alpha}-\delta_{p\alpha}) + (1- j_{\alpha}+\delta_{p\alpha})b_{\alpha})(\sum_{\beta} a_{\beta}j_{\beta} + (1- j_{\beta})b_{\beta})}\right) \nonumber \\
&=& \frac{a_{\mu}a_{\nu}}{W_{\bm{j},\bm{j}}}\left(\frac{1}{W_{\bm{i}+\hat{\bm{\mu}},\bm{i}+\hat{\bm{\mu}}}} + \frac{1}{W_{\bm{i}+\hat{\bm{\nu}},\bm{i}+\hat{\bm{\nu}}}}\right)
\end{eqnarray}
Detailed balance then implies that
\begin{eqnarray}
\frac{\pi_{\bm{j}}}{\pi_{\bm{i}}} &=& \frac{Q^{2}_{\bm{j},\bm{i}}}{Q^{2}_{\bm{i},\bm{j}}} = \frac{b_{\mu}}{a_{\mu}}\frac{b_{\nu}}{a_{\nu}}\frac{W_{\bm{j},\bm{j}}}{W_{\bm{i},\bm{i}}} \nonumber \\
&=& \frac{b_{\mu}}{a_{\mu}}\frac{b_{\nu}}{a_{\nu}}\left( 1 + \frac{(a_{\mu} + a_{\nu}) - (b_{\mu}  + b_{\nu})}{-W_{\bm{i},\bm{i}}}\right)\nonumber\\
&\leq& \frac{c+2}{c}*r^{-2} \textrm{ if } c\neq0
\end{eqnarray}
Since $\pi_{\hat{\bm{\mu}}} = \pi_{\bm{0}}\frac{b_{\mu}(a_{\mu} +
\sum_{\nu\neq\mu}b_{\nu})}{(b_{\mu}
+\sum_{\nu\neq\mu}b_{\nu})a_{\mu}}<\pi_{\bm{0}}$ we can deduce that
for any state $\bm{i}$ with $c$ bonds broken, with $k-1 \geq c \geq 1$, the
invariant probability $\pi_{\bm{i}} \leq c*r^{-c+1}\pi_{\bm{0}}$. Let,
$\bm{l}$ be the state with $k-1$ bonds broken for which
$\pi$ is maximized. The choice of matrix elements imposed by detailed balance
implies $Q_{\bm{l},\bm{b}} \geq 1/(^{d}_{k-1})$. Also, since lemma
\ref{Qbounds} implies $\sum_{\mu} \pi_{\hat{\bm{\mu}}}\leq
\pi_{\bm{0}}(1+d/r)$ we get for all values of $k$: \begin{eqnarray}
\frac{\pi_{\bm{b}}}{\pi_{\bm{l}}}&=&
\frac{Q_{\bm{b},\bm{l}}}{Q_{\bm{l},\bm{b}}}\leq
\frac{(d-k+1)(^{d}_{k-1})}{(k-1)r}\nonumber \\ \Rightarrow
\tilde{\pi}_{\bm{b}}&\leq& \frac{\pi_{\bm{b}}}{\pi_{\bm{0}}} \leq
(d-k+1)(^{d}_{k-1})r^{-k+1}
\end{eqnarray}
Finally, using the lower bound on $E[T_{\bm{b}\bm{0}}]$ computed in preceding theorem we get
\begin{eqnarray}
E[T_{\bm{b}\bm{0}}]&\geq & 2\frac{1 - \left(k/2(1 + d/r)^{k}+2\right)\tilde{\pi}_{\bm{b}}}{\tilde{\pi}_{\bm{b}}}\nonumber \\
&\geq& P(d,k)*r^{k-1} \textbf{ $\forall$ } r>r_{0}
\end{eqnarray}
where $P(d,k)=\frac{2}{(d-k+1)(^{d}_{k-1})}\left(1-(1+(k*2^{k-2})^{-1})\frac{k(d-k+1)(^{d}_{k-1})}{(2/d)^{k-1}}\right)$ and $r_{0}=d$. 
\end{proof}

\end{document}